\documentclass[onefignum,onetabnum]{siamart190516}

\usepackage{tikz}
\usetikzlibrary{3d}
\usetikzlibrary {decorations.pathreplacing}

\usepackage{amssymb}
\usepackage{tikz}
\usetikzlibrary{3d}
\usetikzlibrary {decorations.pathreplacing}
\usepackage{mlmath}
\usepackage{amsmath}
\usepackage{mathrsfs}
\usepackage{algorithm}
\usepackage{algpseudocode}
\usepackage{hyperref}
\newcommand{\mrm}[1]{\mathrm{#1}} % redefine mathrm
\newcommand{\V}[1]{\boldsymbol{#1}} %# vector 
\newcommand{\grad}{{\nabla}} %gradient

\makeatletter
\def\cref@override@label@type#1{}
\makeatother

\usepackage{lipsum}
\usepackage{amsfonts}
\usepackage{graphicx}
\usepackage{epstopdf}
\ifpdf
  \DeclareGraphicsExtensions{.eps,.pdf,.png,.jpg}
\else
  \DeclareGraphicsExtensions{.eps}
\fi

\newsiamremark{remark}{Remark}
\newsiamremark{hypothesis}{Hypothesis}
\crefname{hypothesis}{Hypothesis}{Hypotheses}
\newsiamthm{claim}{Claim}

\headers{}{Z. Gan, X. Gao and Y. Li}

\title{An $\fO(N)$ quasi-Ewald splitting method for nanoconfined electrostatics \thanks{Submitted to the editors DATE.
\funding{This work is supported by the Natural Science Foundation of China Grant Nos. 12201146 (Z. G.) and 12401667~(Y. L.); Natural Science Foundation of Guangdong (Grant No.2023A1515012197); and Guangzhou-HKUST(GZ) joint research project (Grant No.2023A03J0003 and 2024A03J0606).}
}}

\author{Zecheng Gan\thanks{Thrust of Advanced Materials and Guangzhou Municipal Key Laboratory of Materials Informatics, The Hong Kong University of Science and Technology (Guangzhou), Guangdong, China; and Department of Mathematics, The Hong Kong University of Science and Technology, Hong Kong SAR, China (\email{zechenggan@ust.hk}).}
\and Xuanzhao Gao\thanks{Corresponding author. Center for Computational Mathematics, Flatiron Institute, New York, NY, USA (\email{xgao@flatironinstitute.org}).}
\and Yuqing Li\thanks{School of Mathematical Sciences,  Key Laboratory of MEA \& Shanghai Key Laboratory of PMMP,  East China Normal University, Shanghai 200241, China (\email{liyq@math.ecnu.edu.cn})}
}

\usepackage{amsopn}
\usepackage { graphicx,epstopdf } % <- Preamble 
\usepackage[caption=false] { subfig }

\ifpdf
\hypersetup{
  pdftitle={An $\fO(N)$ quasi-Ewald splitting method for nanoconfined electrostatics}  
  pdfauthor={Z. Gan, X. Gao, and Y. Li}
}
\fi

\begin{document}

\maketitle

% REQUIRED
%% Abstract
\begin{abstract}
% Nanocofined systems, characterized by macroscopic dimensions in $xy$ and nano-scaled in $z$, hold fundamental significance across diverge areas such as materials science and soft matter physics.%, allowing efficient numerical simulations for the dynamics of charged particles under dielectric nanoconfinement.
Simulating the dynamics of charged particles in quasi-two-dimensional (quasi-2D) nanoconfined systems presents a significant computational challenge due to the long-range nature of electrostatic interactions and the geometric anisotropy.
To address this, we introduce a novel \emph{quasi-Ewald splitting} strategy tailored for particle-based simulations in such geometry.
Our splitting strategy seamlessly integrates a collection of advanced numerical techniques, including optimal quadrature rules [L. N. Trefethen, $SIAM~Rev.$ 64(1)(2022), pp.132-150], fast pairwise kernel summation methods [S. Jiang and L. Greengard, $Commun.~ Comput.~Phys.$ 31(1)(2022), pp.1-26], and the random batch method with importance sampling in $k$-space [S. Jin, L. Li, Z. Xu et al., $SIAM~J.~Sci.~Comput.$ 43(4)(2021), pp.B937-B960]. 
The resulting algorithm achieves an $\fO(N)$ overall computational complexity, where $N$ denotes the total number of confined particles. 
Simulations of several prototype systems validate the accuracy and efficiency of our method. 
Furthermore, we present numerical observations specifically related to nanoconfined charged many-body systems, highlighting phenomena such as dielectric boundary effects, anisotropic diffusion, and the structure of the electrical double layer (EDL) under conditions of charge asymmetry.
\end{abstract}

%% Keywords



\section{Introduction}
\label{sec...intro}

In recent years, the modeling, simulation, and experimental investigation of quasi-two-dimensional (quasi-2D) systems have gained significant attention across various disciplines~\cite{benjamin1997molecular, chen2006pinned, mccann2006landau, mazars2011long, giovambattista2012computational, moreno2021quantum,yang2023physics,yu2023characterizing}. 
Typically, such  systems possess a nano-sized longitudinal thickness in the $z$ direction, achieved through confinement; they are bulk-like and modeled as doubly periodic in the transverse $xy$ directions, hence endowed with an inherent \emph{multi-scale} nature.
Among quasi-2D systems, charged many-body systems under dielectric nanoconfinement are of particular interest to us, such as graphene~\cite{novoselov2004electric}, metal dichalcogenide monolayers~\cite{kumar2012tunable}, semi-conductors and solid-state battery systems~\cite{li2023recent}.
A key observation in these systems is the \emph{dielectric confinement effect}~\cite{hong1992dielectric,katan2019quantum}, arising from the discontinuity in dielectric properties at the boundaries,  leading to distinct electrostatic screening behaviors.  
This effect plays a pivotal role in governing the properties of charged systems   and   has been implicated in a wide range of phenomena,  including ion transport in nanochannels~\cite{antila2018dielectric}, the structural organization of polymer brushes~\cite{yuan2020structure}, and the pattern formation observed in confined  dipolar systems~\cite{wang2019dielectric, gao2024broken}. 
Despite significant advancements in simulation techniques over the past few decades, achieving accurate and efficient representations of the dielectric confinement effect remains a formidable challenge~\cite{gao2024broken}. 
Moreover, the intrinsic multi-scale nature of quasi-2D systems, combined with the long-range interactions inherent to electrostatic forces,   presents profound difficulties for both theoretical modeling and numerical simulation approaches.

Linear-scaling algorithms have been extensively developed for the efficient simulation of isotropic, non-confined $N$-particle systems under periodic or free-space boundary conditions. 
These approaches primarily follow three strategies: (1) Fourier spectral methods, (2) tree-based methods, and (3) random batch sampling methods. 
Fourier spectral methods combine kernel splitting with the Fast Fourier Transform (FFT) to achieve $\mathcal{O}(N\log N)$ complexity, including the Particle-Mesh Ewald (PME)~\cite{darden1993particle} method and the u-series method~\cite{predescu2020u}.
Tree-based methods, such as the Fast Multipole Method (FMM)~\cite{greengard1987fast, greengard1988rapid,fmm2, fmm3, fmm4, fmm6, fmm7, fmm8} and the dual-space multilevel kernel-splitting (DMK) framework~\cite{jiang2025dual}, utilize hierarchical data structures to reach $\mathcal{O}(N)$ complexity. 
Random batch sampling methods, such as the random batch Ewald (RBE)~\cite{jin2021random} and random batch sum-of-Gaussians (RBSOG)~\cite{liang2023random}, leverage importance sampling in Fourier space to provide $\mathcal{O}(N)$ stochastic acceleration, which is particularly effective for capturing ensemble averages in molecular dynamics simulations~\cite{jin2020random, jin2021convergence, jin2022randombatch}. 
While these methods are robust for isotropic Coulomb systems, quasi-2D systems with confined geometries introduce unique challenges due to reduced symmetry, necessitating either significant modifications to existing algorithms or the development of new methodologies.

Various methods have been developed to reduce the $\mathcal{O}(N^2)$ complexity inherent in quasi-2D confined systems. 
For systems with a homogeneous dielectric background, early contributions include the Lekner summation-based MMM2D method~\cite{arnold2002mmm2d}, which achieves $\mathcal{O}(N^{5/3})$ complexity. 
More recent developments include the Ewald3DC method~\cite{yeh1999ewald}, the spectral Ewald method~\cite{lindbo2012fast}, periodic FMMs~\cite{yan2018flexibly, pei2023fast}, SoEwald2D~\cite{gan2024fast}, and the fast sum-of-Gaussians method~\cite{Gao2025}. 
These approaches typically leverage the Fast Fourier Transform (FFT), Fast Multipole Methods (FMM), or random batch sampling to reach $\mathcal{O}(N\log N)$ or $\mathcal{O}(N)$ complexity. 
However, the requirement of a homogeneous dielectric background limits their general applicability. 
To address sharp dielectric interfaces, one approach involves solving for the induced charge density or surface potential through various optimization techniques~\cite{boda2004computing, tyagi2010iterative, jadhao2012simulation, jadhao2013variational, Barros2014Efficient, Barros2014Dielectric}. 
Alternatively, the Image Charge Method (ICM) can be integrated with established electrostatic solvers~\cite{tyagi2007icmmm2d, tyagi2008electrostatic, dos2015electrolytes, yuan2021particle, liang2020harmonic, liang2022hsma, gan2025random}, and its error can be controlled accurately~\cite{gao2025accurate}.
Despite their potential, these methods still face significant challenges in strongly confined quasi-2D systems. 
For instance, ICM-based methods may require a prohibitive number of image charges to achieve convergence~\cite{gao2025accurate}, while boundary integral-based approaches often become ill-conditioned in quasi-2D geometries~\cite{Barros2014Efficient}, complicating the balance between computational accuracy and efficiency.

In this work, we propose an~$\mathcal{O}(N)$ simulation algorithm, namely the quasi-Ewald method (QEM), for doubly periodic charged particle systems with dielectric mismatches at the confinement boundaries.
To overcome the difficulties caused by the geometry of quasi-2D systems, we propose a tailored quasi-Ewald splitting strategy, which divides the original Coulomb kernel into short-range real space and long-range Fourier space components.
Through the Dirichlet-to-Neumann map, the Green's function for quasi-2D nanoconfined electrostatics is derived,  
one thus arrives at an analytic and fast convergent lattice summation formula for this class of systems.
We further develop a set of efficient and accurate methods to numerically evaluate these components.
We apply analytic reduction combined with a tailored Gauss quadrature rule~\cite{trefethen2022exactness} to numerically evaluate the Hankel transform of 0-th order~\cite{piessens2000hankel} in the short range components, and develop a \textit{separation via sorting} technique~\cite{jiang2021approximating, gan2024fast} to boost the computation of the pairwise summation in the long range components.
Finally, by introducing the \textit{random batch sampling}~\cite{jin2020random, jin2021random, jin2021randomquantum, golse2019random, jin2021convergence, liang2022superscalability, li2020random, jin2022mean} in $k$-space, we achieve linear complexity, making QEM suitable for molecular dynamics simulations of large-scale, nanoconfined quasi-2D systems.
Notably, QEM is \emph{mesh-free} and does not require the introduction of any image charge to handle dielectric mismatch at boundaries.
Consequently, this approach is capable of addressing challenges posed by extremely thin systems with polarizable boundaries, which are challenging to address with existing approaches.
%%%%%%%%%%%%%%%%%%%%%%%%%%%%%%%%%%%%%%%%%%%%%%%%%%%%%%%%%%%%
%%%%%%%%%%%%%%%%%%%%%%%%%%%%%%%%%%%%%%%%%%%%%%%%%%%%%%%%%%%%
%%%%%%%%%%%%%%%%%%%%%%%%%%%%%%%%%%%%%%%%%%%%%%%%%%%%%%%%%%%%
%%%%%%%%%%%%%%%%%%%%%%%%%%%%%%%%%%%%%%%%%%%%%%%%%%%%%%%%%%%%
%%%%%%%%%%%%%%%%%%%%%%%%%%%%%%%%%%%%%%%%%%%%%%%%%%%%%%%%%%%%
%%%%%%%%%%%%%%%%%%%%%%%%%%%%%%%%%%%%%%%%%%%%%%%%%%%%%%%%%%%%
%%%%%%%%%%%%%%%%%%%%%%%%%%%%%%%%%%%%%%%%%%%%%%%%%%%%%%%%%%%%
\section{The model: concepts and mathematical preliminaries}\label{section....Preliminaries}
%%%%%%%%%%%%%%%%%%%%%%%%%%%%%%%%%%%%%%%%%%%%%%%%%%%%%%%%%%%%
%%%%%%%%%%%%%%%%%%%%%%%%%%%%%%%%%%%%%%%%%%%%%%%%%%%%%%%%%%%%
%%%%%%%%%%%%%%%%%%%%%%%%%%%%%%%%%%%%%%%%%%%%%%%%%%%%%%%%%%%%

In this section, we introduce the key model concepts and mathematical notations that will be used throughout this paper. 
First, Fig.~\ref{fig:box} provides a schematic representation of the quasi-2D charged many-body system under confinement.  
We note that quasi-2D systems~\cite{mazars2011long} are essentially 3D systems, with a characteristic length scale in $z$ much smaller than that in the $xy$ plane. 
Physically, these systems model charged many-body interactions with dielectric confinement effects~\cite{dos2017simulations,gao2024broken}.

To accurately compute electrostatic interactions in quasi-2D systems, they are usually modeled as doubly periodic in the $xy$ directions to mimic a bulk-like environment, while incorporating appropriate dielectric confinement interfaces and boundary conditions in the $z$ direction~\cite{maxian2021fast,liang2020harmonic}. 
Specifically, we consider a simulation box with dimensions~$\left(L_x, L_y, L_z\right)$ in $\mathbb R^3$, where dielectric confinement interfaces are positioned at~$z = 0$ and~$ z = L_z$. 
As shown in Fig.~\ref{fig:box}, the two dielectric interfaces divide~$\sR^3$ into three distinct subdomains, denoted as~$\Omega_{\mrm{u}}$,~$\Omega_{\mrm{c}}$ and~$\Omega_{\mrm{d}}$ from top to bottom, respectively. 
The dielectric permittivity is defined as a piecewise constant in the three subdomains,
%%%%%%%%%%%%%%%%%%%%%%%%%%%%%%%%%%%%%%%%%%%%%%%%%%%%%%%%%%%%
\begin{equation}\label{eq:sanwich}
    \epsilon(\V{r})=\left\{
        \begin{array}{ll}
        \epsilon_{\mrm{u}}, & {\V{r} \in \Omega_{\mrm{u}}}\\
        \epsilon_{\mrm{c}}, & {\V{r} \in \Omega_{\mrm{c}}}\\
        \epsilon_{\mrm{d}}, & {\V{r} \in \Omega_{\mrm{d}}}
    \end{array} \right.\;,
\end{equation}
%%%%%%%%%%%%%%%%%%%%%%%%%%%%%%%%%%%%%%%%%%%%%%%%%%%%%%%%%%%%
where~$\epsilon_u$,~$\epsilon_c$ and~$\epsilon_d$ are positive material-specific  constants.

%%%%%%%%%%%%%%%%%%%%%%%%%%%%%%%%%%%%%%%%%%%%%%%%%%%%%%%%%%%%
%%%%%%%%%%%%%%%%%%%%%%%%%%%%%%%%%%%%%%%%%%%%%%%%%%%%%%%%%%%%

% code to generate the position of the particles
% function rand_string(n)
%     xyz = [rand(n) * 4, rand(n) * 3, rand(n) * 4]
%     st = ["$(xyz[1][i])/$(xyz[2][i])/$(xyz[3][i])" for i in 1:n]
%     return reduce((x, y) -> x * ", " * y, st)
% end

\begin{figure}[htbp!]
    \centering
    \begin{tikzpicture}[scale=0.6]

        \draw[->] (-5, 0, 0) -- (9, 0, 0) node[right] {$x$}; % x-axis
        \draw[->] (0, -1, 0) -- (0, 4, 0) node[above] {$z$}; % y-axis
        \draw[->] (0, 0, -1) -- (0, 0, 7) node[above] {$y$}; % z-axis

        \fill[orange, opacity=0.2] (-4.5, 0, -0.5) -- (8.5, 0, -0.5) -- (8.5, 0, 4.5) -- (-4.5, 0, 4.5);
        % \fill[orange, opacity=0.2] (-4.5, 0, 4.5) -- (8.5, 0, 4.5) -- (8.5, -0.5, 4.5) -- (-4.5, -0.5, 4.5);
        % \fill[orange, opacity=0.2] (8.5, 0, -0.5) -- (8.5, 0, 4.5) -- (8.5, -0.5, 4.5) -- (8.5, -0.5, -0.5);

        \draw[thick] (0, 0, 0) -- (4, 0, 0) -- (4, 3, 0) -- (0, 3, 0) -- cycle; % Bottom face

        \foreach \x/\y/\z in {1.1681915835663497/0.4784386449304706/3.0114784670813646, 1.37700874310584/0.4564910752408716/1.723302872827872, 3.730994750135867/1.2345031904085388/2.103105645365449, 1.827895374334192/2.07357017874381/3.746941345635766, 3.066558538444058/1.5385426337573604/1.7619877149348904, 1.323901301301564/1.5495075401772844/0.6031052215686215, 3.756938726146066/2.4640736386255653/3.077635065930065, 2.8772957887165047/2.6997418470234025/2.9329762207627343, 3.2340462748359378/0.05840416805411408/3.426026606528974, 2.6025039709080175/2.9218252913656833/0.57717965442315666} {
            \draw[fill=red] (\x, \y, \z) circle (0.1);
        }

        \foreach \x/\y/\z in {2.6498617350978773/2.150049613814989/1.5048602780413822, 3.089169032276023/1.0562600341829351/3.3127362964215057, 2.5980246280095605/0.7722521776014004/3.9662770148938087, 1.1349256153410514/2.870534190907561/3.169316637530182, 3.1201708293697656/0.9091412698002154/1.6441954803083996, 1.9206578784987878/2.6240285163641914/1.8820001242058, 2.21640094284661/0.839090054913301/0.3248449650128413, 1.7055625414319642/1.0502926700849304/2.1025386222981464, 1.9224549177172139/1.8402480747416887/3.7853267350243893, 0.6942437596034523/1.2454445897320752/1.4473021778716602} {
            \draw[fill=blue] (\x, \y, \z) circle (0.1);
        }

        \foreach \dx/\dy in {-4/0, 4/0} {
            \foreach \x/\y/\z in {1.1681915835663497/0.4784386449304706/3.0114784670813646, 1.37700874310584/0.4564910752408716/1.723302872827872, 3.730994750135867/1.2345031904085388/2.103105645365449, 1.827895374334192/2.07357017874381/3.746941345635766, 3.066558538444058/1.5385426337573604/1.7619877149348904, 1.323901301301564/1.5495075401772844/0.6031052215686215, 3.756938726146066/2.4640736386255653/3.077635065930065, 2.8772957887165047/2.6997418470234025/2.9329762207627343, 3.2340462748359378/0.05840416805411408/3.426026606528974, 2.6025039709080175/2.9218252913656833/0.57717965442315666} {
            \draw[draw=red, fill=white] (\x + \dx, \y + \dy, \z) circle (0.1);
        }

            \foreach \x/\y/\z in {2.6498617350978773/2.150049613814989/1.5048602780413822, 3.089169032276023/1.0562600341829351/3.3127362964215057, 2.5980246280095605/0.7722521776014004/3.9662770148938087, 1.1349256153410514/2.870534190907561/3.169316637530182, 3.1201708293697656/0.9091412698002154/1.6441954803083996, 1.9206578784987878/2.6240285163641914/1.8820001242058, 2.21640094284661/0.839090054913301/0.3248449650128413, 1.7055625414319642/1.0502926700849304/2.1025386222981464, 1.9224549177172139/1.8402480747416887/3.7853267350243893, 0.6942437596034523/1.2454445897320752/1.4473021778716602} {
                \draw[draw=blue, fill=white] (\x + \dx, \y + \dy, \z) circle (0.1);
            }
        }

        \draw[thick] (0, 0, 4) -- (4, 0, 4) -- (4, 3, 4) -- (0, 3, 4) -- cycle; % Top face
        \draw[thick] (0, 0, 0) -- (0, 0, 4); % Front left
        \draw[thick] (4, 0, 0) -- (4, 0, 4); % Front right
        \draw[thick] (4, 3, 0) -- (4, 3, 4); % Back right
        \draw[thick] (0, 3, 0) -- (0, 3, 4); % Back left

        % Draw dashed boxes for all 8 neighboring boxes
        \foreach \dx/\dy in {-4/0, 4/0} {
            \draw[thick, dashed] (\dx - 0.5, 0, \dy) -- (4.5 + \dx, 0, \dy);
            \draw[thick, dashed] (4 + \dx, 0, \dy) -- (4 + \dx, 3, \dy);
            \draw[thick, dashed] (4.5 + \dx, 3, \dy) -- (\dx - 0.5, 3, \dy);
            \draw[thick, dashed] (\dx, 0, \dy) -- (\dx, 3, \dy);

            \draw[thick, dashed] (\dx - 0.5, 0, 4 + \dy) -- (4.5 + \dx, 0, 4 + \dy);
            \draw[thick, dashed] (4 + \dx, 0, 4 + \dy) -- (4 + \dx, 3, 4 + \dy);
            \draw[thick, dashed] (4.5 + \dx, 3, 4 + \dy) -- (\dx - 0.5, 3, 4 + \dy);
            \draw[thick, dashed] (\dx, 0, 4 + \dy) -- (\dx, 3, 4 + \dy);

            \draw[thick, dashed] (\dx, 0, \dy - 0.5) -- (\dx, 0, 4.5 + \dy); % Front left
            \draw[thick, dashed] (4 + \dx, 0, \dy - 0.5) -- (4 + \dx, 0, 4.5 + \dy); % Front right
            \draw[thick, dashed] (4 + \dx, 3, \dy - 0.5) -- (4 + \dx, 3, 4.5 + \dy); % Back right
            \draw[thick, dashed] (\dx, 3, \dy - 0.5) -- (\dx, 3, 4.5 + \dy); % Back left
        }

        \fill[green, opacity=0.2] (-4.5, 3, -0.5) -- (8.5, 3, -0.5) -- (8.5, 3, 4.5) -- (-4.5, 3, 4.5) -- cycle;
        % \fill[green, opacity=0.2] (-4.5, 3, -0.5) -- (8.5, 3, -0.5) -- (8.5, 3.5, -0.5) -- (-4.5, 3.5, -0.5) -- cycle;
        % \fill[green, opacity=0.2] (-4.5, 3, -0.5) -- (-4.5, 3, 4.5) -- (-4.5, 3.5, 4.5) -- (-4.5, 3.5, -0.5) -- cycle;

        \node at (2, 0, 5) {$L_x$};
        \node at (4.6, 0, 2) {$L_y$};
        \node at (-0.5, 0.5, 0) {$L_z$};

        \draw[decorate, decoration={brace, amplitude=10pt, mirror}] (8.5, 0.2, 0) -- (8.5, 2.8, 0) node[midway,xshift=20] {$\Omega_{\mrm{c}}$};
        \draw[decorate, decoration={brace, amplitude=10pt, mirror}] (8.5, 3.2, 0) -- (8.5, 4.2, 0) node[midway,xshift=20] {$\Omega_{\mrm{u}}$};
        \draw[decorate, decoration={brace, amplitude=10pt, mirror}] (8.5, -1.2, 0) -- (8.5, -0.2, 0) node[midway,xshift=20] {$\Omega_{\mrm{d}}$};

    \end{tikzpicture}
    \caption{
        A schematic of the quasi-2D nanoconfined particle system.
        The filled circles denote the charged particles confined within the central simulation box, while the empty circles represent their periodic replicas in the $xy$ directions.
        The green and orange regions highlight the sharp dielectric interfaces, i.e., $\partial \Omega_c \cap \partial \Omega_{u}$ and $\partial \Omega_c \cap \partial \Omega_{d}$, respectively.
        Note that the system is periodic in both $x$ and $y$, but for clarity, only the periodic replicas along the~$x$ axis are sketched here.
    }
    \label{fig:box}
\end{figure}
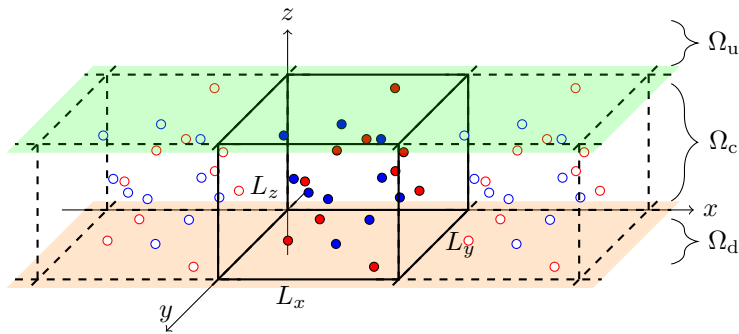
%%%%%%%%%%%%%%%%%%%%%%%%%%%%%%%%%%%%%%%%%%%%%%%%%%%%%%%%%%%%
%%%%%%%%%%%%%%%%%%%%%%%%%%%%%%%%%%%%%%%%%%%%%%%%%%%%%%%%%%%%

Next, we describe molecular dynamics (MD) simulations of quasi-2D charged systems in a canonical (NVT) ensemble.
Consider~$N$ charged particles inside the simulation box with positions~$\V{r}_i = (x_i, y_i, z_i)$, carrying charges~$q_i$ ($i=1,2,\ldots, N$), and satisfying the overall-charge-neutral condition (if not, the electrostatic interaction will be ill-defined under doubly periodic boundary conditions~\cite{gan2024fast}):
%%%%%%%%%%%%%%%%%%%%%%%%%%%%%%%%%%%%%%%%%%%%%%%%%%%%%%%%%%%%
\begin{equation}\label{eq:totalneutral}
    \sum_{i = 1}^N q_i = 0\;.
\end{equation}
%%%%%%%%%%%%%%%%%%%%%%%%%%%%%%%%%%%%%%%%%%%%%%%%%%%%%%%%%%%%
Then MD simulation of the $N$ charged particles in a canonical  ensemble essentially requires numerically integrating the following $6N$-dimensional Newton’s equations with Langevin thermostat~\cite{frenkel2023understanding},
\begin{equation}\label{eq:NewtonEquation}
    \begin{split}
        \mrm{d} \V{r}_i &= \V{v}_i \mrm{d} t\;,\\
        m_i \mrm{d} \V{v}_i &= \V{F}_i \mrm{d} t + \mrm{d} \V{\eta}_i \;,
    \end{split}
\end{equation}
%%%%%%%%%%%%%%%%%%%%%%%%%%%%%%%%%%%%%%%%%%%%%%%%%%%%%%%%%%%%
where~$\V{F}_i$ represents the force exerted on the $i$-th particle, and~$\D \V{\eta}_i=-\gamma \V{v}_i \D t +  \sqrt{\frac{2\gamma}{\beta}} \D\V{W}_i$ represents the coupling with the environment (heat bath), where $\V{W}_i$ are i.i.d. Wiener processes, $\gamma$ the reciprocal characteristic time associated with the thermostat, and $\beta=\frac{1}{k_{\mathrm{B}}T}$.
It is understood that for Langevin dynamics, the damping term $-\gamma \V{v}_i \D t$ and thermal fluctuation term $\sqrt{\frac{2\gamma}{\beta}} \D\V{W}_i$ are balanced by the fluctuation-dissipation relationship, thus approximating a NVT ensemble.
Finally, the force~$\V{F}_i$ is computed via~$\V{F}_i:= - \grad_{\V{r}_i} U$, where $U$ is the potential energy of the system. 
Specifically, for electrostatic energy of quasi-2D systems, the potential energy $U$ reads,
\begin{equation}
    U = \frac{1}{2} {\sum_{\V{m}}}^\prime \sum_{i, j = 1}^{N} q_i q_j G(\V{r}_i;\V{r}_j+ \V{L}_{\V{m}} )\;,
\end{equation}
where~$\V{m}=(m_x, m_y, 0)$ with $ (m_x, m_y)\in \mathbb{Z}^2$, $\V{L}_{\V{m}} := (m_x L_x, m_y L_y, 0)$ denotes the quasi-2D lattice vectors, and the lattice sum
${\sum\limits_{\V{m}}}^\prime$ is defined such that the singular contribution in the bare Coulomb interaction for $i=j,~\vm=\vzero$ is excluded. 
Finally,  we define $G(\V{r};\V{r}^\prime)$ as the Green's function  in a quasi‑2D geometry, governed by Poisson’s equation subject to dielectric interface conditions:  
%%%%%%%%%%%%%%%%%%%%%%%%%%%%%%%%%%%%%%%%%%%%%%%%%%%%%%%%%%%%
\begin{equation}\label{eq...text...Poisson}
    \left \{
    \begin{array}{ll}
        - \grad_{\V{r}} \cdot\left[ \epsilon(\V{r}) \grad_{\V{r}} G(\V{r};\V{r}^\prime) \right] = \delta (\V{r} - \V{r}^\prime), & \text{for}~\V r \in \mathbb{R}^3 \;, \\
%%%%%%%%%%%%%%%%%%%%%%%%%%%        
        G(\V{r};\V{r}^\prime) |_{-} = G(\V{r};\V{r}^\prime) |_{+}, & \text{on}~\partial\Omega_{\mrm{c}}\;, \\
%%%%%%%%%%%%%%%%%%%%%%%%%%%        
        \epsilon_{\mrm{c}} \partial_{z} G(\V{r};\V{r}^\prime) |_{-} = \epsilon_{\mrm{u}} \partial_{z} G(\V{r};\V{r}^\prime) |_{+}, & \text{on}~\partial \Omega_{\mrm{c}} \cap \partial \Omega_{\mrm{u}}\;,     \\
%%%%%%%%%%%%%%%%%%%%%%%%%%%            
\epsilon_{\mrm{c}} \partial_{z} G(\V{r};\V{r}^\prime) |_{-} = \epsilon_{\mrm{d}} \partial_{z} G(\V{r};\V{r}^\prime) |_{+}, & \text{on}~\partial \Omega_{\mrm{c}} \cap \partial \Omega_{\mrm{d}}\;,\\
%%%%%%%%%%%%%%%%%%%%%%%%%%%  
G(\V{r};\V{r}^\prime) \to 0, & \text{as}~{r} \to \infty\;,
\end{array}\right.
\end{equation}
%%%%%%%%%%%%%%%%%%%%%%%%%%%%%%%%%%%%%%%%%%%%%%%%%%%%%%%%%%%%
where $r:=\Norm{\vr}=\sqrt{x^2+y^2+z^2}$, and $\Norm{\cdot}$ represents the Euclidean norm.
%together with free-space boundary condition $G(\V{r},~\V{r}^\prime) \to 0$ as $z\to\pm \infty$ and periodic boundary conditions in $xy$.

The ICM can be applied to analytically solve the above set of equations, which expresses the Green's function as an infinite series of image charges, constructed through an iterative process of image-charge reflections~\cite{jackson2021classical,liang2022hsma}. 
Note that after each reflection, the magnitude of image charges is weakened by a dimensionless factor $\gamma_{\mrm{u}}$ (or $\gamma_{\mrm{d}}$), measuring the dielectric contrasts for the upper and lower dielectric interfaces, defined as
%%%%%%%%%%%%%%%%%%%%%%%%%%%%%%%%%%%%%%%%%%%%%%%%%%%%%%%%%%%%
\begin{equation}
    \gamma_{\mrm{u}} := \frac{\epsilon_{\mrm{c}} - \epsilon_{\mrm{u}}}{\epsilon_{\mrm{c}} + \epsilon_{\mrm{u}}},~~
    \gamma_{\mrm{d}} := \frac{\epsilon_{\mrm{c}} - \epsilon_{\mrm{d}}}{\epsilon_{\mrm{c}} + \epsilon_{\mrm{d}}}\;.
\end{equation}
%%%%%%%%%%%%%%%%%%%%%%%%%%%%%%%%%%%%%%%%%%%%%%%%%%%%%%%%%%%%
Since the dielectric constants are all positive, consequently, ${\gamma_{\mrm{u}}}$ and $\gamma_{\mrm{d}}$ satisfy that $\Abs{\gamma_{\mrm{u}}}<1$ and $\Abs{\gamma_{\mrm{d}}}<1$,
hence the following infinite series  ICM representation of the Green's function~$G(\V{r};\V{r}^\prime)$ is guaranteed to be convergent:
\begin{equation}\label{eq:ICM}
    G(\V{r}; \V{r}^\prime) = \frac{1}{4 \pi \epsilon_{\mrm{c}}} \left[ \frac{1}{\Norm{\V{r} - \V{r}^\prime}} + \sum_{l = 1}^\infty \left( \frac{\gamma_l^+}{\Norm{\V{r} - \V{r}_{+}^{\prime(l)}}} + \frac{\gamma_l^-}{\Norm{\V{r} - \V{r}_{-}^{ \prime(l)}}} \right) \right]\;,
\end{equation}
%%%%%%%%%%%%%%%%%%%%%%%%%%%%%%%%%%%%%%%%%%%%%%%%%%%%%%%%%%%%
where the factors for the $l$-th level image charges read $\gamma_l^{+}=\gamma_{d}^{\lceil l/2 \rceil} \gamma_{u}^{\lfloor l/2 \rfloor}$ and $\gamma_l^{-} = \gamma_{d}^{\lfloor l/2 \rfloor} \gamma_{u}^{\lceil l/2 \rceil}$. Here the superscript $+(-)$ indicates the location of image charges in $\Omega_{\textrm {u}}\left(\Omega_{\textrm {d}}\right)$. The notation $\lceil x\rceil(\lfloor x\rfloor)$ represents the ``ceil" (``floor") function, and the position of the $l$th-level images are given by,
\begin{equation}
\boldsymbol{r}_{\pm}^{\prime(l)}=\left(x^{\prime}, y^{\prime},(-1)^l z^{\prime} \pm l L_z\right)\;.
\end{equation}
The above ICM representation of the Green's function provides theoretical benchmark for systems considered here. However, it is important to note that due to the long-range nature of electrostatic interactions, Eq.~\eqref{eq:ICM} is not well suited for large-scale MD simulations of quasi-2D charged systems due to its high computational complexity, especially for strongly-confined systems (i.e., when $L_z\ll L_{x,y}$).

Finally, we revisit the idea of the Ewald splitting technique for efficient calculations of the electrostatic interactions in a quasi-2D geometry.
% :
% \begin{equation}
% 	\delta(\bm{r})=\left[\delta(\bm{r})-(\delta\ast\tau)(\bm{r})\right]+(\delta\ast\tau)(\bm{r}),
% \end{equation}
% where   $\ast$ stands for  the convolution operator, and $\tau(\bm{r})$ is the screening function.
%In the standard Ewald splitting~\cite{ewald1921berechnung}, $\tau$ is chosen to be  Gaussian.
For quasi-2D systems \emph{without} dielectric interface conditions, the so-called Ewald2D method~\cite{parry1975electrostatic, heyes1977molecular, de1979electrostatic} has been developed, which decomposes the interaction energy $U$ as
\begin{equation}
    U= U_s+U_l\;,
\end{equation}
where the short-range and long-range components are given by
\begin{equation}\label{eq:Ewald2D}
    \begin{split}
        U_s &= \frac12 \sum_{i, j = 1}^N q_i q_j {\sum_{\V{m}}}^{\prime} \frac{\mrm{erfc}(\sqrt{\alpha} \Norm{\V{r}_i - \V{r}_j + \V{L_m}})}{\Norm{\V{r}_i - \V{r}_j + \V{L_m}}} - \sqrt{\frac{\alpha}{{\pi}}} \sum_{i = 1}^{N} q_i^2 \;, \\
        U_l &= \frac12 \sum_{i, j = 1}^N q_i q_j {\sum_{\V{m}}}  \frac{\mrm{erf}(\sqrt{\alpha} \Norm{\V{r}_i - \V{r}_j + \V{L_m}})}{\Norm{\V{r}_i - \V{r}_j + \V{L_m}}}\;,%=U_l^0 + \sum_{k > 0} U_l^k\;,
    \end{split}
\end{equation}
respectively. Note that~$\mrm{erf}(x)$ is the error function 
\begin{equation}
    \mrm{erf}(x) := \frac{2}{\sqrt{\pi}} \int_0^{x} e^{-u^2} \mrm{d} u \;,
\end{equation}
and $\mrm{erfc}(x) := 1- \mrm{erf}(x)$, $\alpha >0$ is the splitting parameter that can be adjusted to maximize performance in MD simulations~\cite{frenkel2023understanding}. 
The computational efficiency can be enhanced via Ewald decomposition due to the fact that, (a) $U_s$ decays fast in real space and is short-ranged, so it can be efficiently evaluated via truncation; and (b) $U_l$ is still long-ranged, but the interaction becomes smooth even when $\V r_i\to \V r_j$ and $\V m=\V 0$, so that it converges spectrally in Fourier space.
Consequently, based on \cite[Lemma 2.4]{gan2024fast}, the potential $U_l$ can be transformed into the following summation form and evaluated in the $k$-space,
\begin{equation}
 U_l  =  U_l^{\bm{0}} +\sum_{\V{k} \neq \bm{0}} U_l^{\V{k}}  \;,   
\end{equation}
where,
	\begin{equation}\label{eq:philk}
	 U_l^{\V{k}} := \frac{\pi}{2L_x L_y} \sum_{i,j = 1}^N q_i q_{j} \frac{e^{\mathrm{i} \V{k} \cdot \V{\rho}_{ij}}}{k} \left[\xi^{+}(k, z_{ij})+\xi^{-}(k, z_{ij})\right],
	\end{equation}
	with $\bm{\rho}_{ij}:=(x_{i}-x_{j},y_{i}-y_{j})$, $z_{ij}:=\Norm{z_{i}-z_{j}}$,  $k:=\Norm{\vk}$, and 
	\begin{equation}\label{eq::9}
		\xi^{\pm}(k, z_{ij}) := e^{{\pm k z_{ij}}} \mrm{erfc} \left( \frac{k}{2 \sqrt{\alpha}} \pm \sqrt{\alpha} z_{ij}\right)\;,
	\end{equation} 
	\begin{equation}\label{eq:phil0}
		 U_l^{\bm{0}}  := -\frac{ \pi}{L_x L_y} \sum_{j = 1}^N q_iq_{j} \left[ {z_{ij}} \mrm{erf}(\sqrt{\alpha} {z_{ij}}) + \frac{e^{-\alpha z_{ij}^2}}{\sqrt{\alpha \pi}}  \right]\;.
	\end{equation} 
%%%%%%%%%%%%%%%%%%%%%%%%%%%%%%%%%%%%%%%%%%%%%%%%%%%%%%%%%%%%%%%%%
While the Ewald2D summation accelerates numerical computations for the interaction energy $U$, it has several significant drawbacks: (a) even with an optimal choice of $\alpha$, computing   $U$ for an $N$-particle system requires a computational complexity of $O\left(N^2\right)$, which is still very expensive for large particle numbers; (b) the functions $\xi^{ \pm}\left(k, z_{i j}\right)$ are ill-conditioned, as they grow exponentially with
$k z_{{i j}}$, resulting in catastrophic error cancellation in numerical computations; and (c) the Ewald2D formula is applicable only to homogeneous systems without dielectric interface conditions. 
These challenges underscore the need to develop novel numerical methods that can simultaneously handle long-range interactions, quasi-2D geometry, and dielectric interface conditions, while ensuring both numerical stability and computational efficiency.
%%%%%%%%%%%%%%%%%%%%%%%%%%%%%%%%%%%%%%%%%%%%%%%%%%%%%%%%%%%%
%%%%%%%%%%%%%%%%%%%%%%%%%%%%%%%%%%%%%%%%%%%%%%%%%%%%%%%%%%%%
%%%%%%%%%%%%%%%%%%%%%%%%%%%%%%%%%%%%%%%%%%%%%%%%%%%%%%%%%%%%
%%%%%%%%%%%%%%%%%%%%%%%%%%%%%%%%%%%%%%%%%%%%%%%%%%%%%%%%%%%%
%%%%%%%%%%%%%%%%%%%%%%%%%%%%%%%%%%%%%%%%%%%%%%%%%%%%%%%%%%%%
%%%%%%%%%%%%%%%%%%%%%%%%%%%%%%%%%%%%%%%%%%%%%%%%%%%%%%%%%%%%
\section{The quasi-Ewald splitting for Green's function}
%%%%%%%%%%%%%%%%%%%%%%%%%%%%%%%%%%%%%%%%%%%%%%%%%%%%%%%%%%%%
%%%%%%%%%%%%%%%%%%%%%%%%%%%%%%%%%%%%%%%%%%%%%%%%%%%%%%%%%%%%
%%%%%%%%%%%%%%%%%%%%%%%%%%%%%%%%%%%%%%%%%%%%%%%%%%%%%%%%%%%%

In this section, we present the quasi-Ewald splitting, a novel kernel decomposition strategy tailored for quasi-2D geometry. 

%Based on this, we further develop efficient algorithms that achieve linear complexity, allowing efficient simulations of charged many-body systems under dielectric confinement.
%%%%%%%%%%%%%%%%%%%%%%%%%%%%%%%%%%%%%%%%%%%%%%%%%%%%%%%%%%%%
%%%%%%%%%%%%%%%%%%%%%%%%%%%%%%%%%%%%%%%%%%%%%%%%%%%%%%%%%%%%
%\subsection{Quasi-Ewald splitting and Green's Function in the Reciprocal Space}
%%%%%%%%%%%%%%%%%%%%%%%%%%%%%%%%%%%%%%%%%%%%%%%%%%%%%%%%%%%%
%%%%%%%%%%%%%%%%%%%%%%%%%%%%%%%%%%%%%%%%%%%%%%%%%%%%%%%%%%%%
%%%%%%%%%%%%%%%%%%%%%%%%%%%%%%%%%%%%%%%%%%%%%%%%%%%%%%%%%%%%
First, consider a point source located at $\V{r}^\prime=(x^{\prime}, y^{\prime},z^{\prime})$. The quasi-Ewald splitting decomposes the Dirac delta function as
%%%%%%%%%%%%%%%%%%%%%%%%%%%%%%%%%%%%%%%%%%%%%%%%%%%%%%%%%%
\begin{equation}\label{eq...text...Quasi-EwaldMethod...QuasiEwaldSplitting}
\begin{aligned}
    \delta(\V{r}-\V{r}^\prime) &= \left[\delta (\V{r}-\V{r}^\prime) - {\left( \frac{\alpha}{\pi} \right)} e^{-\alpha \Norm{\vrho-\vrho^\prime}^2} \delta(z-z^\prime) \right]\\&~~ + {\left( \frac{\alpha}{\pi} \right)} e^{-\alpha \Norm{\vrho-\vrho^\prime}^2} \delta(z-z^\prime)\;,
    \end{aligned}
\end{equation}
%%%%%%%%%%%%%%%%%%%%%%%%%%%%%%%%%%%%%%%%%%%%%%%%%%%%%%%%%%%%
where $\vr:=(\vrho, z):=(x,y,z)$,   $\vrho:=(x,y)$, and $\vrho^\prime:=(x^\prime,y^\prime)$.
A 3D illustration of the quasi-Ewald splitting is shown in Fig.~\ref{fig:QuasiEwaldSplitting} for charged particles in a unit simulation cell under quasi-2D confinement.
In Eq.~\eqref{eq...text...Quasi-EwaldMethod...QuasiEwaldSplitting}, the first term on the right-hand side (RHS) represents a point source \emph{screened} by a Gaussian density centered at $\vr^\prime$, spread in the $xy$ plane to align with the quasi-2D cylindrical symmetry. 
The second term on the RHS compensates for this by adding back the additional Gaussian contribution to restore the original Delta function.
We further define the two densities on the RHS separately as, 
%%%%%%%%%%%%%%%%%%%%%%%%%%%%%%%%%%%%%%%%%%%%%%%%%%%%%%%%%%%%
\begin{equation}\label{eq...text...Kernels}
    \begin{split}
        \sigma_{s}(\V{r}; \vr^\prime) &:= \delta (\V{r}-\V{r}^\prime) - {\left( \frac{\alpha}{\pi} \right)}e^{-\alpha \Norm{\vrho-\vrho^\prime}^2} \delta(z-z^\prime)\;,\\
        \sigma_{l}(\V{r};\vr^\prime) &:= {\left( \frac{\alpha}{\pi} \right)} e^{-\alpha \Norm{\vrho-\vrho^\prime}^2}\delta(z-z^\prime)\;.
    \end{split}
\end{equation}
% %%%%%%%%%%%%%%%%%%%%%%%%%%%%%%%%%%%%%%%%%%%%%%%%%%%%%%%%%%%%
%  Consequently, the total charge density in the system can be expressed as 
% %%%%%%%%%%%%%%%%%%%%%%%%%%%%%%%%%%%%%%%%%%%%%%%%%%%%%%%%%%%%
% \begin{equation}
%     \sigma(\vr) = \sum_{\V{m}} \sum_{i = 1}^{N} q_i \left[\sigma_{s}(\V{r} - \V{r}_i + \V{L_m})  +  \sigma_{s}(\V{r} - \V{r}_i + \V{L_m})\right]\;,
% \end{equation}
%%%%%%%%%%%%%%%%%%%%%%%%%%%%%%%%%%%%%%%%%%%%%%%%%%%%%%%%%%%%
%%%%%%%%%%%%%%%%%%%%%%%%%%%%%%%%%%%%%%%%%%%%%%%%%%%%%%%%%%%%
%%%%%%%%%%%%%%%%%%%%%%%%%%%%%%%%%%%%%%%%%%%%%%%%%%%%%%%%%%%%
%%%%%%%%%%%%%%%%%%%%%%%%%%%%%%%%%%%%%%%%%%%%%%%%%%%%%%%%%%%%
\begin{figure}
    \centering
    \begin{tikzpicture}[scale = 0.6]

        \def\r{0.2}

        \draw[thick, dashed] (0, 0, 0) -- (4, 0, 0);
        \draw[thick] (4, 0, 0) -- (4, 3, 0);
        \draw[thick] (4, 3, 0) -- (0, 3, 0);
        \draw[thick, dashed] (0, 3, 0) -- (0, 0, 0);

        \foreach \x/\y/\z in {1.1681915835663497/0.4784386449304706/3.0114784670813646, 1.37700874310584/0.4564910752408716/1.723302872827872, 3.730994750135867/1.2345031904085388/2.103105645365449} {
            \shade[ball color=red!60, opacity=1] (\x, \y, \z) circle (\r);
        }

        \foreach \x/\y/\z in {1.827895374334192/2.07357017874381/3.746941345635766, 3.066558538444058/1.5385426337573604/1.7619877149348904, 1.323901301301564/1.5495075401772844/0.6031052215686215} {
            \shade[ball color=blue!60, opacity=1] (\x, \y, \z) circle (\r);
        }

        \draw[thick] (0, 0, 4) -- (4, 0, 4) -- (4, 3, 4) -- (0, 3, 4) -- cycle; % Top face
        \draw[thick, dashed] (0, 0, 0) -- (0, 0, 4); % Front left
        \draw[thick] (4, 0, 0) -- (4, 0, 4); % Front right
        \draw[thick] (4, 3, 0) -- (4, 3, 4); % Back right
        \draw[thick] (0, 3, 0) -- (0, 3, 4); % Back left

        \foreach \x/\y/\z in {1.1681915835663497/0.4784386449304706/3.0114784670813646, 1.37700874310584/0.4564910752408716/1.723302872827872, 3.730994750135867/1.2345031904085388/2.103105645365449} {
            % draw an circle plane laying in the xz plane
            \fill[fill=blue!60, opacity=0.3] (\x + 8, \y, \z) ellipse (0.7 and 0.3);
            \shade[ball color=red!60, opacity=1] (\x + 8, \y, \z) circle (\r);
        }

        \foreach \x/\y/\z in {1.827895374334192/2.07357017874381/3.746941345635766, 3.066558538444058/1.5385426337573604/1.7619877149348904, 1.323901301301564/1.5495075401772844/0.6031052215686215} {
            \fill[fill=red!60, opacity=0.3] (\x + 8, \y, \z) ellipse (0.7 and 0.3);
            \shade[ball color=blue!60, opacity=1] (\x + 8, \y, \z) circle (\r);
        }

        \foreach \x/\y/\z in {1.1681915835663497/0.4784386449304706/3.0114784670813646, 1.37700874310584/0.4564910752408716/1.723302872827872, 3.730994750135867/1.2345031904085388/2.103105645365449} {
            % draw an circle plane laying in the xz plane
            \fill[fill=red!60, opacity=0.3] (\x + 16, \y, \z) ellipse (0.7 and 0.3);
        }

        \foreach \x/\y/\z in {1.827895374334192/2.07357017874381/3.746941345635766, 3.066558538444058/1.5385426337573604/1.7619877149348904, 1.323901301301564/1.5495075401772844/0.6031052215686215} {
            \fill[fill=blue!60, opacity=0.3] (\x + 16, \y, \z) ellipse (0.7 and 0.3);
        }

        \foreach \x/\y/\z in {8/0/0, 16/0/0} {
            \draw[thick, dashed] (\x, \y, \z) -- (4 + \x, \y, \z);
            \draw[thick] (4 + \x, \y, \z) -- (4 + \x, 3 + \y, \z);
            \draw[thick] (4 + \x, 3 + \y, \z) -- (\x, 3 + \y, \z);
            \draw[thick, dashed] (\x, \y, \z) -- (\x, 3 + \y, \z);
            \draw[thick] (\x, \y, 4 + \z) -- (4 + \x, \y, 4 + \z) -- (4 + \x, 3 + \y, 4 + \z) -- (\x, 3 + \y, 4 + \z) -- cycle;
            \draw[thick, dashed] (\x, \y, \z) -- (\x, \y, 4 + \z);
            \draw[thick] (4 + \x, \y, \z) -- (4 + \x, \y, 4 + \z);
            \draw[thick] (4 + \x, 3 + \y, \z) -- (4 + \x, 3 + \y, 4 + \z);
            \draw[thick] (\x, 3 + \y, \z) -- (\x, 3 + \y, 4 + \z);
        }

        \node at (5.2, 1, 0) {$=$};
        \node at (13.2, 1, 0) {$+$};

        \node at (1.5, 3.5, 0) {(a)};
        \node at (9.5, 3.5, 0) {(b)};
        \node at (17.5, 3.5, 0) {(c)};

    \end{tikzpicture}
    \caption{
        An illustration of the Quasi-Ewald splitting strategy in a unit cell.
        (a) Charged particles are distributed in the simulation box, cations and anions are represented by red and blue spheres, respectively.
        Sub-figures (b) and (c) show how Quasi-Ewald splitting works. 
        The discs represent 2D Gaussian charge density clouds in the $xy$ planes, which screen the point sources.
        One thus splits the original problem into two sub-problems, one with charge neutral particles so that the interactions are short-ranged, the other with smooth charge density which fits into the cylindrical symmetry of quasi-2D geometry so that it can be solved rapidly in the reciprocal space.
    }
    \label{fig:QuasiEwaldSplitting}
\end{figure}
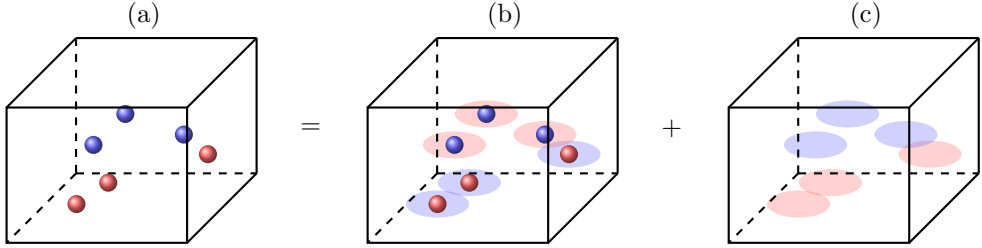
%%%%%%%%%%%%%%%%%%%%%%%%%%%%%%%%%%%%%%%%%%%%%%%%%%%%%%%%%%%%

Next, we define the following \emph{quasi-2D Fourier transform}, which will be extensively used throughout this paper.
\begin{definition}\label{defi:2dfourier}
Given a function $f(\vrho, z)$, its quasi-2D Fourier
transform is defined by
  \begin{equation} 
\Hat{f}(\vk, z):=\int_{\sR^2}f(\vrho,z)e^{-\mathrm{i} \V{k} \cdot \V{\rho}} \D \vrho \;.
\end{equation}
The function  $f(\vrho, z)$ can be recovered from $\Hat{f}(\vk, z)$ by its corresponding inverse quasi-2D Fourier transform
  \begin{equation} 
{f}(\vrho, z)=\frac{1}{4\pi^2}\int_{\sR^2}\Hat{f}(\vk, z)e^{\mathrm{i} \V{k} \cdot \V{\rho}} \D \vk \;.
\end{equation}
\end{definition}
% \begin{equation}\label{eq...text...FourierTransformKernels}
%     \begin{split}   
%  \hat{\sigma}_{s}(\vk,z) &=\delta(z)\left(1-e^{-\frac{k^2}{4\alpha}}\right)\;,\\
%  \hat{\sigma}_{l}(\vk,z) &=\delta(z)e^{-\frac{k^2}{4\alpha}}\;.
%     \end{split}
% \end{equation}
%%%%%%%%%%%%%%%%%%%%%%%%%%%%%%%%%%%%%%%%%%%%%%%%%%%%%%%%%%%%
%%%%%%%%%%%%%%%%%%%%%%%%%%%%%%%%%%%%%%%%%%%%%%%%%%%%%%%%%%%%
%%%%%%%%%%%%%%%%%%%%%%%%%%%%%%%%%%%%%%%%%%%%%%%%%%%%%%%%%%%%
%%%%%%%%%%%%%%%%%%%%%%%%%%%%%%%%%%%%%%%%%%%%%%%%%%%%%%%%%%%%
Taking the quasi-2D Fourier transform for $\sigma_{s}$ and $\sigma_{l}$ in Eq.~\eqref{eq...text...Kernels}, we obtain their counterparts $\Hat{\sigma}_{s}$ and $\Hat{\sigma}_{l}$, summarized in the following proposition.
\begin{proposition}\label{prop:1}
Given a fixed point     $\V{r}^\prime=(x^{\prime}, y^{\prime},z^{\prime})$ with $\V{\rho}^\prime=(x^{\prime}, y^{\prime})$,  the  non-zero modes of the 2D-Fourier
transform  of  $ \sigma_{s}(\V{r};\vr^\prime)$ and   $\sigma_{l}(\V{r};\vr^\prime)$  respectively read
\begin{align*}
 \Hat{\sigma}_{s}(\vk,z;z^\prime) &= \delta(z-z^\prime) e^{-\mathrm{i} \V{k} \cdot \V{\rho}^\prime}\left(1-e^{-\frac{k^2}{4\alpha}}\right)\;,  \\
 \Hat{\sigma}_{l}(\vk,z;z^\prime) &= \delta(z-z^\prime) e^{-\mathrm{i} \V{k} \cdot \V{\rho}^\prime}e^{-\frac{k^2}{4\alpha}}\;,
\end{align*}   
where $\vk\neq\vzero$, and the zeroth mode reads
\begin{align*}
 \Hat{\sigma}_{s}(\vzero,z;z^\prime) &= 0\;,  \\
 \Hat{\sigma}_{l}(\vzero,z;z^\prime) &= \delta(z-z^\prime)\;.
\end{align*}   
\end{proposition}
%%%%%%%%%%%%%%%%%%%%%%%%%%%%%%%%%%%%%%%%%%%%%%%%%%%%%%%%%%%%
%%%%%%%%%%%%%%%%%%%%%%%%%%%%%%%%%%%%%%%%%%%%%%%%%%%%%%%%%%%%
\begin{proof}
By direct application of Definition~\ref{defi:2dfourier}, the integrals can be evaluated analytically, hence we omit the details.
\end{proof}
%%%%%%%%%%%%%%%%%%%%%%%%%%%%%%%%%%%%%%%%%%%%%%%%%%%%%%%%%%%%
%%%%%%%%%%%%%%%%%%%%%%%%%%%%%%%%%%%%%%%%%%%%%%%%%%%%%%%%%%%%
%%%%%%%%%%%%%%%%%%%%%%%%%%%%%%%%%%%%%%%%%%%%%%%%%%%%%%%%%%%%
Now, given a source point $\vr_0:=(\vrho_0,z_0):=(x_0, y_0, z_0)\in\Omega_{\mrm{c}}$, the Green's function  $G(\V{r}; \V{r_0})$ (as defined in Eq.~\ref{eq...text...Poisson}) can be solved in $k$-space by using Prop~\ref{prop:1} and the Dirichlet-to-Neumann map. 
The resulting expressions are summarized in the following Theorem \ref{Proposition...GreensFunction}, whose proof can be found in Appendix.
%  \begin{equation}\label{eq...text...RepPoisson}
%     \left \{
%     \begin{array}{ll}
%         - \grad_{\V{r}} \cdot\left[ \epsilon(\V{r}) \grad_{\V{r}} G(\V{r},~\V{r}_0) \right] = \delta (\V{r} - \V{r}_0), & \V r \in \mathbb{R}^3 \;, \\
% %%%%%%%%%%%%%%%%%%%%%%%%%%%        
%         G(\V{r},~\V{r}_0) |_{-} = G(\V{r},~\V{r}_0) |_{+}, & \text{on}~\partial\Omega_{\mrm{c}}\;, \\
% %%%%%%%%%%%%%%%%%%%%%%%%%%%        
%         \epsilon_{\mrm{c}} \partial_{z} G(\V{r},~\V{r}_0) |_{-} = \epsilon_{\mrm{u}} \partial_{z} G(\V{r},~\V{r}_0) |_{+}, & \text{on}~\partial \Omega_{\mrm{c}} \cap \partial \Omega_{\mrm{u}}\;,     \\
% %%%%%%%%%%%%%%%%%%%%%%%%%%%            
% \epsilon_{\mrm{c}} \partial_{z} G(\V{r},~\V{r}_0) |_{-} = \epsilon_{\mrm{d}} \partial_{z} G(\V{r},~\V{r}_0) |_{+}, & \text{on}~\partial \Omega_{\mrm{c}} \cap \partial \Omega_{\mrm{d}}\;,\\
% %%%%%%%%%%%%%%%%%%%%%%%%%%%  
% G(\V{r},~\V{r}_0) \to 0, & \text{as}~{r} \to \infty\;.
% \end{array}\right.
% \end{equation}
%%%%%%%%%%%%%%%%%%%%%%%%%%%%%%%%%%%%%%%%%%%%%%%%%%%%%%%%%%%%
%We shall employ the  to derive  the Green's function  $G(\V{r}, \V{r_0})$    in the reciprocal space. 
%%%%%%%%%%%%%%%%%%%%%%%%%%%%%%%%%%%%%%%%%%%%%%%%%%%%%%%%%%%%
%%%%%%%%%%%%%%%%%%%%%%%%%%%%%%%%%%%%%%%%%%%%%%%%%%%%%%%%%%%%
\begin{theorem}\label{Proposition...GreensFunction}
    Let us define
%%%%%%%%%%%%%%%%%%%%%%%%%%%%%%%%%%%%%%%%%%%%%%%%%%%%%%%%%%%%
    \begin{equation}
        \Hat{G}(\vk, z; z_0) = \int_{\mathbb{R}^2} G(\V{r}; \V{r_0}) e^{-\mathrm{i} \V{k} \cdot \V{\rho}} \D \vrho\;,
    \end{equation}
%%%%%%%%%%%%%%%%%%%%%%%%%%%%%%%%%%%%%%%%%%%%%%%%%%%%%%%%%%%%     
then for $\vk=\vzero$,
%%%%%%%%%%%%%%%%%%%%%%%%%%%%%%%%%%%%%%%%%%%%%%%%%%%%%%%%%%%%   
\begin{equation}\label{eq...hat_G_solution...Zeroth}
        \hat{G}(\vzero, z; z_0) = \frac{\Abs{z-z_0}}{2 \epsilon_{\mathrm{c}}}\;,
\end{equation}
%%%%%%%%%%%%%%%%%%%%%%%%%%%%%%%%%%%%%%%%%%%%%%%%%%%%%%%%%%%%        
and for $\vk\neq \vzero$,
%%%%%%%%%%%%%%%%%%%%%%%%%%%%%%%%%%%%%%%%%%%%%%%%%%%%%%%%%%%%     
    \begin{equation}\label{eq...hat_G_solution...Non-zero}
    \begin{aligned}
\hat{G}(\vk, z; z_0) &=  \frac{ e^{-\mathrm{i} \V{k} \cdot \V{\rho}_0}}{2k\epsilon_{\mrm{c}}\left(1-\gamma_{\mrm{u}}\gamma_{\mrm{d}}e^{-2kL_z}\right)} \sum_{p=1}^4\Gamma_pe^{-ka_p(z;z_0)}\;, %\Big[\exp(-k\Abs{z-z_0})\\
% &~~+\gamma_{\mrm{d}}\exp(-k(z+z_0))+\gamma_{\mrm{u}}\exp(-k(2L_z-z-z_0))\\
% &~~+\gamma_{\mrm{u}}\gamma_{\mrm{d}}\exp(-k(2L_z-\Abs{z-z_0}))\Big]\;.
    \end{aligned}     
    \end{equation}
%%%%%%%%%%%%%%%%%%%%%%%%%%%%%%%%%%%%%%%%%%%%%%%%%%%%%%%%%%%%         
where 
\begin{align*}
 \Gamma_{1:4}&:=\left[1,\gamma_{\mrm{d}},\gamma_{\mrm{u}},\gamma_{\mrm{u}}\gamma_{\mrm{d}}\right]\;,\\
 a_{1:4}(z;z_0)&:=\left[\Abs{z-z_0}, z+z_0,2L_z-z-z_0,2L_z-\Abs{z-z_0}\right]\;.
\end{align*}
\end{theorem}
\begin{remark}\label{Remark...Nonzero+larger}
 (a) The condition $\abs{\gamma_u \gamma_d} < 1$  ensures the uniqueness of   the quasi-2D Green's function;
 (b) All the components in $a_{1:4}(z;z_0)$ have positive values; (c) The overall charge neutrality condition Eq.~\eqref{eq:totalneutral} is needed so that for $k=0$, the infinite boundary condition can be met for $N$-particle systems.
\end{remark}
%%%%%%%%%%%%%%%%%%%%%%%%%%%%%%%%%%%%%%%%%%%%%%%%%%%%%%%%%%%%
%%%%%%%%%%%%%%%%%%%%%%%%%%%%%%%%%%%%%%%%%%%%%%%%%%%%%%%%%%%%
Based on the quasi-Ewald splitting we introduced in Eqs.~\eqref{eq...text...Quasi-EwaldMethod...QuasiEwaldSplitting}-\eqref{eq...text...Kernels} for the delta source, that is, $\delta(\V{r}-\V{r}_0)=\sigma_s  (\V{r};\V{r}_0)+\sigma_l  (\V{r};\V{r}_0)\;,$
%%%%%%%%%%%%%%%%%%%%%%%%%%%%%%%%%%%%%%%%%%%%%%%%%%%%%%%%%%%%
we shall find the corresponding decomposition for the Green's function as
%%%%%%%%%%%%%%%%%%%%%%%%%%%%%%%%%%%%%%%%%%%%%%%%%%%%%%%%%%%%
\begin{equation}
    G(\V{r}; \V{r_0})= G_s(\V{r}; \V{r_0})+ G_l(\V{r}; \V{r_0})\;,
\end{equation}
%%%%%%%%%%%%%%%%%%%%%%%%%%%%%%%%%%%%%%%%%%%%%%%%%%%%%%%%%%%%
where $ G_s(\V{r}; \V{r_0})$ and $ G_l(\V{r}; \V{r_0})$ are the electrostatic potentials generated by the charge densities $\sigma_s  (\V{r};\V{r}_0)$ and $\sigma_l  (\V{r};\V{r}_0)$, respectively.
These can be solved analogously by using Prop.~\ref{prop:1}; the resulting expressions are summarized in the following corollary.
%%%%%%%%%%%%%%%%%%%%%%%%%%%%%%%%%%%%%%%%%%%%%%%%%%%%%%%%%%%%
%%%%%%%%%%%%%%%%%%%%%%%%%%%%%%%%%%%%%%%%%%%%%%%%%%%%%%%%%%%%
\begin{corollary}\label{corollary1}
The Green’s function given in Theorem~\ref{Proposition...GreensFunction} admits a quasi‑Ewald splitting, namely,
\begin{equation}
    G(\V{r}; \V{r_0})= G_s(\V{r};\V{r_0})+ G_l(\V{r}; \V{r_0})\;,
\end{equation}
where for any $\vk\neq \vzero$, 
    \begin{align*}
        \hat{G}_s(\vk, z; z_0)&=\hat{G}(\vk, z; z_0)\left(1-e^{-\frac{k^2}{4\alpha}}\right)\;,\\
        \hat{G}_l(\vk, z; z_0)&=\hat{G}(\vk, z; z_0)e^{-\frac{k^2}{4\alpha}}\;,
    \end{align*} 
    and for  $\vk=\vzero$, 
        \begin{align*}
        \hat{G}_s(\vzero, z; z_0)&=0\;,\\
        \hat{G}_l(\vzero, z; z_0)&=\hat{G}(\vzero, z; z_0)\;.
    \end{align*} 
\end{corollary}
\begin{proof}
The proof is straightforward: one can solve for $\hat{G}_s$ and $\hat{G}_l$ by replacing the Dirac delta source with $\sigma_s$ and $\sigma_l$, then using Prop.~\ref{prop:1} and following the same steps as was done in the proof of Thm.~\ref{Proposition...GreensFunction}.
We omit the details.
\end{proof}
%%%%%%%%%%%%%%%%%%%%%%%%%%%%%%%%%%%%%%%%%%%%%%%%%%%%%%%%%%%%
%%%%%%%%%%%%%%%%%%%%%%%%%%%%%%%%%%%%%%%%%%%%%%%%%%%%%%%%%%%%
Finally, one can recover $ G_s(\V{r}; \V{r_0})$ and $ G_l(\V{r}; \V{r_0})$ in real space by using the inverse quasi-2D Fourier transform,
%%%%%%%%%%%%%%%%%%%%%%%%%%%%%%%%%%%%%%%%%%%%%%%%%%%%%%%%%%%%
\begin{equation*}
\begin{aligned}
 G_s(\V{r}; \V{r_0})&=   \frac{1}{4\pi^2}\int_{\sR^2} \hat{G}_s(\vk, z; z_0)e^{\mathrm{i} \V{k} \cdot \V{\rho}} \D \vk= \frac{1}{4\pi^2}\int_{\sR^2} \hat{G}(\vk, z; z_0)\left(1-e^{-\frac{k^2}{4\alpha}}\right)e^{\mathrm{i} \V{k} \cdot \V{\rho}} \D \vk\;,\\
 G_l(\V{r}; \V{r_0})&=   \frac{1}{4\pi^2}\int_{\sR^2} \hat{G}_l(\vk, z; z_0)e^{\mathrm{i} \V{k} \cdot \V{\rho}} \D \vk= \frac{1}{4\pi^2}\int_{\sR^2} \hat{G}(\vk, z; z_0)e^{-\frac{k^2}{4\alpha}}e^{\mathrm{i} \V{k} \cdot \V{\rho}} \D \vk\;.
\end{aligned}    
\end{equation*}
%%%%%%%%%%%%%%%%%%%%%%%%%%%%%%%%%%%%%%%%%%%%%%%%%%%%%%%%%%%%
%%%%%%%%%%%%%%%%%%%%%%%%%%%%%%%%%%%%%%%%%%%%%%%%%%%%%%%%%%%%
Thus the total electrostatic interaction energy $U$ can be expressed as
%%%%%%%%%%%%%%%%%%%%%%%%%%%%%%%%%%%%%%%%%%%%%%%%%%%%%%%%%%%%
\begin{align*}
   U & =   U_s+U_l\;,
\end{align*}
%%%%%%%%%%%%%%%%%%%%%%%%%%%%%%%%%%%%%%%%%%%%%%%%%%%%%%%%%%%%
where
%%%%%%%%%%%%%%%%%%%%%%%%%%%%%%%%%%%%%%%%%%%%%%%%%%%%%%%%%%%%
\begin{equation}\label{eq::U_s}
U_s:=    \frac{1}{2} {\sum_{\V{m}}}^\prime \sum_{i, j = 1}^{N} q_i q_j G_s(\V{r}_i; \V{r}_j+\vL_m)\;, \end{equation}
and 
\begin{equation}\label{eq::U_l}
U_l:=   \frac{1}{2} {\sum_{\V{m}}}\sum_{i, j = 1}^{N} q_i q_j G_l(\V{r}_i; \V{r}_j+\vL_m)\;. 
\end{equation}
In the following sections, we will present a collection of numerical techniques that enable the evaluation of Eqs.~\eqref{eq::U_s}-\eqref{eq::U_l} in particle-based simulations with an optimal computational cost of $\mathcal O(N)$.
%%%%%%%%%%%%%%%%%%%%%%%%%%%%%%%%%%%%%%%%%%%%%%%%%%%%%%%%%%%%
%%%%%%%%%%%%%%%%%%%%%%%%%%%%%%%%%%%%%%%%%%%%%%%%%%%%%%%%%%%%

%%%%%%%%%%%%%%%%%%%%%%%%%%%%%%%%%%%%%%%%%%%%%%%%%%%%%%%%%%%%
%%%%%%%%%%%%%%%%%%%%%%%%%%%%%%%%%%%%%%%%%%%%%%%%%%%%%%%%%%%%
%%%%%%%%%%%%%%%%%%%%%%%%%%%%%%%%%%%%%%%%%%%%%%%%%%%%%%%%%%%%
%%%%%%%%%%%%%%%%%%%%%%%%%%%%%%%%%%%%%%%%%%%%%%%%%%%%%%%%%%%%
%%%%%%%%%%%%%%%%%%%%%%%%%%%%%%%%%%%%%%%%%%%%%%%%%%%%%%%%%%%%
\section{Optimal quadrature rules for short-range interactions}
%%%%%%%%%%%%%%%%%%%%%%%%%%%%%%%%%%%%%%%%%%%%%%%%%%%%%%%%%%%%
%%%%%%%%%%%%%%%%%%%%%%%%%%%%%%%%%%%%%%%%%%%%%%%%%%%%%%%%%%%%
We begin by deriving an optimal quadrature rule for the evaluation of the short-range interaction energy $U_s$. 
As discussed earlier, $G_s$ is the potential field generated by~$\sigma_s$, which represents a source point screened by a flattened Gaussian density surrounding it.
Due to this fact, the leading order of $G_s$ in the far field behaves as a quadrupole, i.e., the potential decays as~$1/r^3$, which is a short-range interaction in $\mathbb R^3$.
Consequently, the infinite lattice summation in Eq.~\eqref{eq::U_s} can be truncated to a finite sum in real space with an inter-particle cutoff distance $r_c$, 
    \begin{equation}\label{eq::U_s_truncated}
        U_s \approx U_{s, *} = {\sum_{\Norm{\bm{\rho}_{ij, m}} \leq r_c}}^\prime q_i q_j G_s(\V{r}_i; \V{r}_j+\vL_m)\;,
    \end{equation}
    where $\bm{\rho}_{ij, m}:=\vrho_i-\vrho_j+ (L_x m_x, L_y m_y)$, indicating that among all possible pairs of particles, only the terms with $\Norm{\bm{\rho}_{ij, m}} \leq r_c$ need to be computed. As such, neighbor-list algorithms can be applied to reduce the cost to $\mathcal{O}(N)$~\cite{frenkel2023understanding} if~$r_c$ is a constant that does not vary with the system size.

The main difficulty in computing Eq.~\eqref{eq::U_s_truncated} lies in efficient and accurate quadrature rules for evaluating $G_s$, which reads,
\begin{align*}
   &G_s(\V{r}_i; \V{r}_j) \\
   =& \frac{1}{8\pi^2\epsilon_{\mrm{c}}}\int_{\sR^2} \frac{e^{\mathrm{i} \V{k} \cdot \V{\rho}_{ij}}}{k\left(1-\gamma_{\mrm{u}}\gamma_{\mrm{d}}e^{-2kL_z}\right)} \left(1-e^{-\frac{k^2}{4\alpha}}\right)\left(\sum_{p=1}^4\Gamma_pe^{-ka_p(z_i;z_j)}\right)\D \vk\;,
\end{align*}
where $\V{\rho}_{ij}:=\vrho_i-\vrho_j$. Due to the cylindrical symmetry, $G_s$ can be further reduced to the following 1D \emph{inverse Hankel transform} format,
\begin{equation}\label{eq::Gs_1dform}
\begin{aligned}
   &G_s(\V{r}_i; \V{r}_j)\\
   =& \frac{1}{8\pi^2\epsilon_{\mrm{c}}}\int_{\sR^2} \frac{e^{\mathrm{i} \V{k} \cdot \V{\rho}_{ij}}}{k\left(1-\gamma_{\mrm{u}}\gamma_{\mrm{d}}e^{-2kL_z}\right)} \left(1-e^{-\frac{k^2}{4\alpha}}\right)\left(\sum_{p=1}^4\Gamma_pe^{-ka_p(z_i;z_j)}\right)\D \vk\;,\\
   =&\frac{1}{4\pi\epsilon_{\mrm{c}}}\int_0^{+\infty} \frac{ e^{-ka_p(z_i;z_j)}}{1-\gamma_{\mrm{u}}\gamma_{\mrm{d}}e^{-2kL_z}} \left(1-e^{-\frac{k^2}{4\alpha}}\right)\fJ_0(\rho_{ij}k)\D k\;,
\end{aligned}
\end{equation}
where $\rho_{ij}:=\Norm{\vrho_{ij}}=\Norm{\vrho_i-\vrho_j}$, and $\fJ_0(\cdot)$ is the zeroth-order Bessel function of the first kind, whose expression reads $\fJ_0(r):=\frac{1}{2\pi} \int_0^{2\pi} e^{\mathrm{i} r\sin \theta} \D \theta\;.$

%%%%%%%%%%%%%%%%%%%%%%%%%%%%%%%%%%%%%%%%%%%%%%%%%%%%%%%%%%%% 

%%%%%%%%%%%%%%%%%%%%%%%%%%%%%%%%%%%%%%%%%%%%%%%%%%%%%%%%%%%%
Consider Eq.~\eqref{eq::Gs_1dform}, without loss of generality, we develop optimal quadrature rules for evaluating the following integral
%%%%%%%%%%%%%%%%%%%%%%%%%%%%%%%%%%%%%%%%%%%%%%%%%%%%%%%%%%%%
\begin{equation}\label{eq...text...InfiniteIntegral}
\int_0^{+\infty} \frac{ e^{-ak}}{1-\gamma_{\mrm{u}}\gamma_{\mrm{d}}e^{-2kL_z}} \left(1-e^{-\frac{k^2}{4\alpha}}\right)\fJ_0(\rho k)\D k\;,    
\end{equation}
%%%%%%%%%%%%%%%%%%%%%%%%%%%%%%%%%%%%%%%%%%%%%%%%%%%%%%%%%%%%
where $a>0$ and $\rho>0$ are constants. 
First, we split Eq.~\eqref{eq...text...InfiniteIntegral} into the summation of two terms,
%%%%%%%%%%%%%%%%%%%%%%%%%%%%%%%%%%%%%%%%%%%%%%%%%%%%%%%%%%%%
\begin{align*}
   & \int_0^{+\infty} \frac{ e^{-ak}}{1-\gamma_{\mrm{u}}\gamma_{\mrm{d}}e^{-2kL_z}} \left(1-e^{-\frac{k^2}{4\alpha}}\right)\fJ_0(\rho k)\D k\\
   =&\underbrace{\int_0^{+\infty} \frac{1}{1-\gamma_{\mrm{u}}\gamma_{\mrm{d}}e^{-2kL_z}}  e^{-ak} \fJ_0(\rho k)\D k}_{\text{Term}~\mathrm{I}}-\underbrace{\int_0^{+\infty} \frac{1}{1-\gamma_{\mrm{u}}\gamma_{\mrm{d}}e^{-2kL_z}}e^{-\frac{k^2}{4\alpha}}  e^{-ak} \fJ_0(\rho k)\D k}_{\text{Term}~\mathrm{II}}\;.
\end{align*}
%%%%%%%%%%%%%%%%%%%%%%%%%%%%%%%%%%%%%%%%%%%%%%%%%%%%%%%%%%%%
Clearly, Term II is   numerically more tractable, as its integrand contains the factor $e^{\frac{-k^2}{4\alpha}}$, which  decays rapidly with increasing 
$k$. 
For term I, we further decompose it into two terms,
%%%%%%%%%%%%%%%%%%%%%%%%%%%%%%%%%%%%%%%%%%%%%%%%%%%%%%%%%%%%
\begin{align*}
&\int_0^{+\infty} \frac{1}{1-\gamma_{\mrm{u}}\gamma_{\mrm{d}}e^{-2kL_z}}  e^{-ak} \fJ_0(\rho k)\D k\\
=&\underbrace{\int_0^{+\infty}   e^{-ak} \fJ_0(\rho k)\D k}_{\text{Term}~\mathrm{I(a)}}+\underbrace{\int_0^{+\infty} \frac{\gamma_{\mrm{u}}\gamma_{\mrm{d}}}{1-\gamma_{\mrm{u}}\gamma_{\mrm{d}}e^{-2kL_z}} e^{-2kL_z}e^{-ak} \fJ_0(\rho k)\D k}_{\text{Term}~\mathrm{I(b)}}\;.
\end{align*}
%%%%%%%%%%%%%%%%%%%%%%%%%%%%%%%%%%%%%%%%%%%%%%%%%%%%%%%%%%%%
Notice that $\left\{e^{-ax}, \frac{1}{\sqrt{a^2+1}}\right\}$ is a Hankel transform pair, namely, $\int_0^{+\infty}e^{-ax}\fJ_0(x)\D x=\frac{1}{\sqrt{a^2+1}}$, thus
%%%%%%%%%%%%%%%%%%%%%%%%%%%%%%%%%%%%%%%%%%%%%%%%%%%%%%%%%%%%
Term $\mathrm{I(a)}$ can be obtained analytically, 
%%%%%%%%%%%%%%%%%%%%%%%%%%%%%%%%%%%%%%%%%%%%%%%%%%%%%%%%%%%%
\begin{equation}\label{eq...text...TermI(A)}
    \int_0^{+\infty}   e^{-ak} \fJ_0(\rho k)\D k=\frac{1}{\sqrt{a^2+\rho^2}}\;.
\end{equation}
%%%%%%%%%%%%%%%%%%%%%%%%%%%%%%%%%%%%%%%%%%%%%%%%%%%%%%%%%%%%
Although $\fJ_0(r)$ is slowly decaying and oscillatory as $r\to+\infty$, both Term $\mathrm{I(b)}$ and Term $\mathrm{II}$ contain rapidly decaying terms $e^{-\frac{k^2}{4\alpha}}$ and $e^{-k L_z}$, allowing efficient calculations using optimal quadrature rules proposed by Trefethen et al. in~\cite{trefethen2022exactness}.
%facilitates numerical computation by truncating the integration region to a finite domain. 
%%%%%%%%%%%%%%%%%%%%%%%%%%%%%%%%%%%%%%%%%%%%%%%%%%%%%%%%%%%%
%%%%%%%%%%%%%%%%%%%%%%%%%%%%%%%%%%%%%%%%%%%%%%%%%%%%%%%%%%%%
% \begin{thm}[Theorem 5.1 in \cite{trefethen2022exactness}]\label{Thm...Truncation}
%  Let $f(x)$ be analytic and bounded for $x \in(-\infty, \infty)$, and suppose $F(x):=\exp \left(-x^2\right) f(x)$ extends to a bounded analytic function in the infinite strip $-a<\operatorname{Im} x<$ a for some $a>0$. Let $L>0$ be fixed, and for each $n \geq 1$, as we denote  \[I_n(F):=\sum_{j=1}^n w_jF(x_j)\;,\] as the estimate of the integral \[I(F):=\int_{-\infty}^{+\infty} F(x)\D x \;,\]  obtained by applying Gauss-Legendre, Clenshaw-Curtis, or equispaced trapezoidal quadrature on the truncated interval $\left[-L n^{1 / 3}, L n^{1 / 3}\right]$. Then for some $C>0$,  the following holds 
%  \begin{equation}
%  \left|I-I_n\right|=\fO\left(\exp \left(-C n^{2 / 3}\right)\right), \quad n \rightarrow \infty\;.  
%  \end{equation}
% \end{thm}
%%%%%%%%%%%%%%%%%%%%%%%%%%%%%%%%%%%%%%%%%%%%%%%%%%%%%%%%%%%%
%%%%%%%%%%%%%%%%%%%%%%%%%%%%%%%%%%%%%%%%%%%%%%%%%%%%%%%%%%%%
For a given truncation parameter $M>0$, the infinite integral can be first approximated by truncation,
%%%%%%%%%%%%%%%%%%%%%%%%%%%%%%%%%%%%%%%%%%%%%%%%%%%%%%%%%%%%
\begin{equation*}
 \int_0^{+\infty} F(k)\D  k \approx \int_0^{M} F(k)\D  k\;,
\end{equation*}
%%%%%%%%%%%%%%%%%%%%%%%%%%%%%%%%%%%%%%%%%%%%%%%%%%%%%%%%%%%% 
and then the integral on $[0, M]$ can be efficiently evaluated using  Gauss-Legendre quadrature. 
We further provide truncation error estimates to determine $M$ in practice, which are summarized in the following theorem.
% For the Term $\mathrm{I(b)}$, its truncation error reads
% %%%%%%%%%%%%%%%%%%%%%%%%%%%%%%%%%%%%%%%%%%%%%%%%%%%%%%%%%%%%
% \begin{equation}\label{eq...text...TruncationError(Term1(b))}
% \begin{aligned}
%     &\Norm{\int_M^{+\infty} \frac{\gamma_{\mrm{u}}\gamma_{\mrm{d}}}{1-\gamma_{\mrm{u}}\gamma_{\mrm{d}}e^{-2kL_z}}  e^{-2kL_z}e^{-ak} \fJ_0(\rho k)\D k}\\
%     \leq &  \frac{\Abs{\gamma_{\mrm{u}}\gamma_{\mrm{d}}}}{1-\max\{0,\gamma_{\mrm{u}}\gamma_{\mrm{d}}\}}\int_M^{+\infty}  e^{-2kL_z}) \Norm{\fJ_0(\rho k)}\D k
% \end{aligned}    
% \end{equation}
%%%%%%%%%%%%%%%%%%%%%%%%%%%%%%%%%%%%%%%%%%%%%%%%%%%%%%%%%%%%
\begin{theorem}\label{thm:truncation_error}
We denote the truncation errors of Term $\mathrm{I(b)}$ and Term $\mathrm{II}$ respectively by
\begin{align*}
  \Delta \mathrm{I}_b(M)&:=\int_M^{+\infty} \frac{\gamma_{\mrm{u}}\gamma_{\mrm{d}}}{1-\gamma_{\mrm{u}}\gamma_{\mrm{d}}e^{-2kL_z}}  e^{-2kL_z}e^{-ak} \fJ_0(\rho k)\D k \;,\\  
  \Delta \mathrm{II}(M)&:=\int_M^{+\infty}  \frac{1}{1-\gamma_{\mrm{u}}\gamma_{\mrm{d}}e^{-2kL_z}}e^{-\frac{k^2}{4\alpha}}  e^{-ak} \fJ_0(\rho k)\D k \;,
\end{align*}
%%%%%%%%%%%%%%%%%%%%%%%%%%%%%%%%%%%%%%%%%%%%%%%%%%%%%%%%%%%%
whose estimates read
%%%%%%%%%%%%%%%%%%%%%%%%%%%%%%%%%%%%%%%%%%%%%%%%%%%%%%%%%%%%
\begin{align}
    \Norm{\Delta \mathrm{I}_b(M)} &\leq  \frac{\Abs{\gamma_{\mrm{u}}\gamma_{\mrm{d}}}}{1-\max\{0,\gamma_{\mrm{u}}\gamma_{\mrm{d}}\}}\frac{1}{\sqrt{\rho L_z}}\mrm{erfc}(\sqrt{2L_zM})\;,\\
    \Norm{\Delta \mathrm{II}(M)} &\leq \frac{\sqrt{\pi\alpha} }{1-\max\{0,\gamma_{\mrm{u}}\gamma_{\mrm{d}}\}} \mrm{erfc}\left(\frac{M}{2\sqrt{\alpha}}\right)\;.
\end{align}
%%%%%%%%%%%%%%%%%%%%%%%%%%%%%%%%%%%%%%%%%%%%%%%%%%%%%%%%%%%%
\end{theorem}
%%%%%%%%%%%%%%%%%%%%%%%%%%%%%%%%%%%%%%%%%%%%%%%%%%%%%%%%%%%%
%%%%%%%%%%%%%%%%%%%%%%%%%%%%%%%%%%%%%%%%%%%%%%%%%%%%%%%%%%%%
\begin{proof}
The proof relies on the properties of the special function $\fJ_0(r)$~\cite{landau2000bessel,olenko2006upper}, details are omitted here and can be found in Appendix. Note that $\mathrm{erfc}(\cdot)$ decays exponentially, ensuring fast convergence in $M$.
%%%%%%%%%%%%%%%%%%%%%%%%%%%%%%%%%%%%%%%%%%%%%%%%%%%%%%%%%%%%
%%%%%%%%%%%%%%%%%%%%%%%%%%%%%%%%%%%%%%%%%%%%%%%%%%%%%%%%%%%%
\end{proof}
%%%%%%%%%%%%%%%%%%%%%%%%%%%%%%%%%%%%%%%%%%%%%%%%%%%%%%%%%%%%
%%%%%%%%%%%%%%%%%%%%%%%%%%%%%%%%%%%%%%%%%%%%%%%%%%%%%%%%%%%%

\begin{figure}[htbp!]
    \centering
    \includegraphics[width = \linewidth]{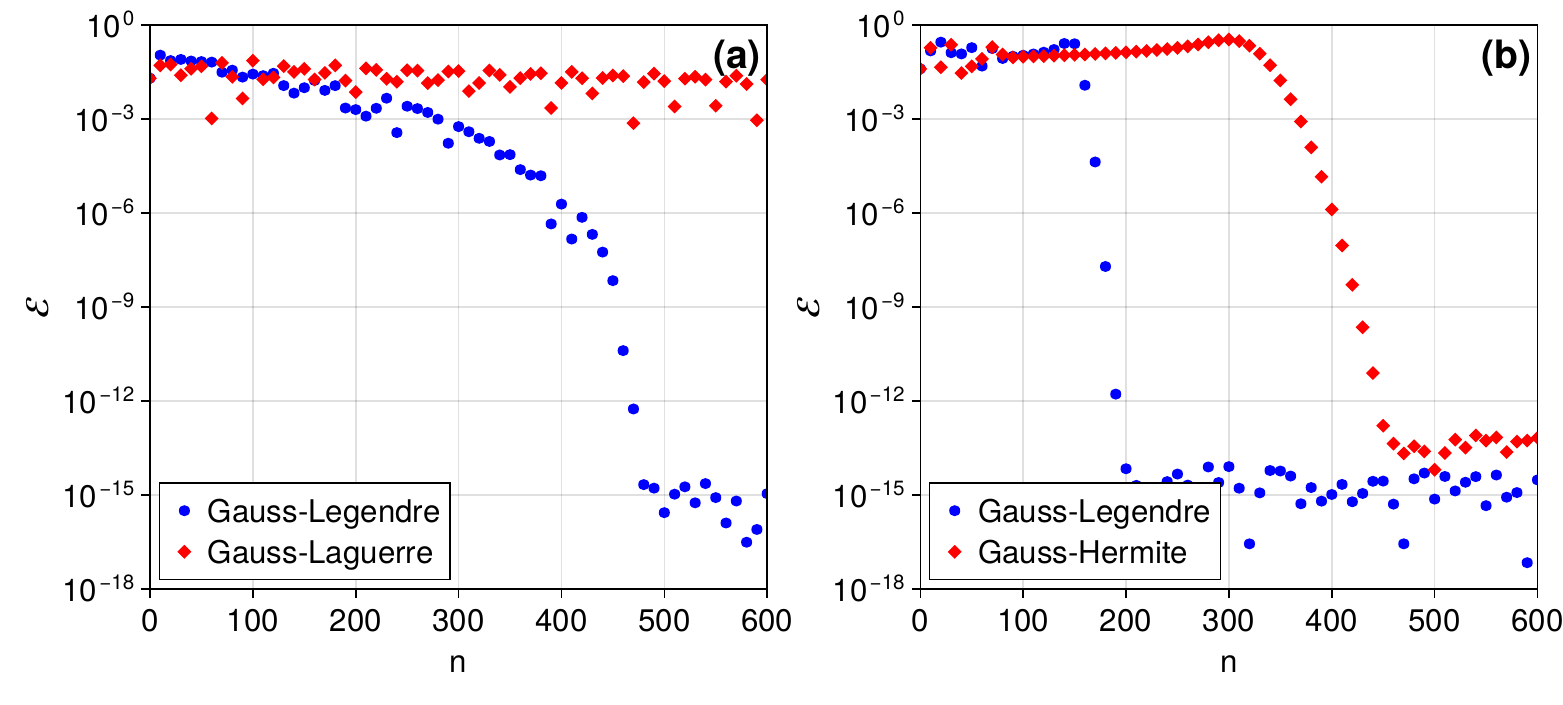}
    \caption{
 %       The relative error of numerical evaluation of integrals~$\int_0^{\infty} \fJ_0(50 x) \exp{(-x)} \mrm{d}x$ and~$\int_{-\infty}^{\infty} J_0(50 x) \exp{(-x^2)} \mrm{d}x$ via (a) Gauss–Laguerre quadrature and (b) Gauss–Hermite quadrature on infinite interval or Gauss-Legendre quadrature on truncated interval $[0, 36]$ and~$[-6, 6]$,~$n$ is the order of the quadrature.
 Absolute error $\mathcal E$ in the evaluation of 
        \(\displaystyle\int_0^{\infty} \fJ_0(50 x) e^{-x} \,\mathrm{d}x\) and 
        \(\displaystyle\int_{-\infty}^{\infty} \fJ_0(50 x) e^{-x^2} \,\mathrm{d}x\) 
        using (a). Gauss--Laguerre~(half interval), and (b). Gauss--Hermite~(full interval) and truncated 
        Gauss--Legendre   quadrature on finite intervals \([0,36]\) and \([-6,6]\), 
        where \(n\) denotes the order of  quadrature.
 }
    \label{fig:Gauss_int}
\end{figure}

\begin{figure}[htbp!]
    \centering
    \includegraphics[width = \linewidth]{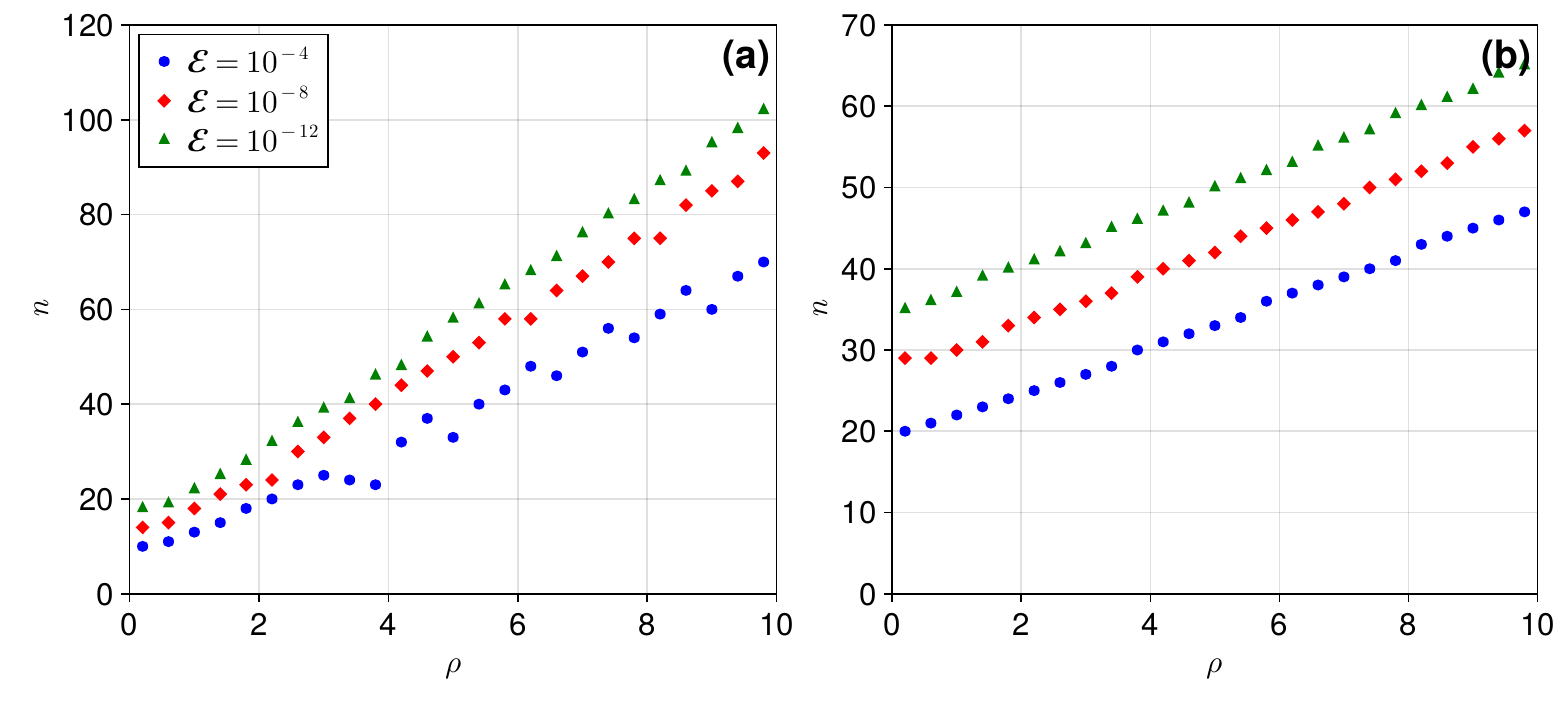}
    \caption{
  %  Order of quadrature~$n$, needed to numerically evaluate (a)~$\int_0^{\infty} J_0(\rho x) \exp{(-x)} \mrm{d}x$ and (b)~$\int_{-\infty}^{\infty} J_0(\rho x) \exp{(-x^2)} \mrm{d}x$ to reach the given  accuracy~$\mathcal{E}$ using the Gauss-Legendre quadrature on truncated interval~$[0, 36]$ and~$[-6, 6]$, respectively, where~$\rho \in [0, 10]$ determine the strength of the oscillation and the lengths of quadrature intervals~$k_{c1}$ and~$k_{c2}$ are determined according to Thm.~\ref{thm:truncation_error}.
   % Quadrature order \(n\) required to achieve accuracy \(\mathcal{E}\) 
   %      using truncated Gauss--Legendre quadrature on \([0,36]\) and \([-6,6]\) for 
   %      (a) \(\int_0^\infty J_0(\rho x) e^{-x} \,\mathrm{d}x\) and 
   %      (b) \(\int_{-\infty}^\infty J_0(\rho x) e^{-x^2} \,\mathrm{d}x\), 
   %      with \(\rho \in [0,10]\) controlling oscillation frequency.
   %      Interval lengths are chosen via Thm.~\ref{thm:truncation_error}.
    Required quadrature order \(n\) vs. oscillation parameter \(\rho \in [0,10]\) 
        for the  prescribed accuracy \(\mathcal{E}\), using truncated Gauss--Legendre 
        quadrature    for 
       (a) \(\int_0^\infty J_0(\rho x) e^{-x} \,\mathrm{d}x\)  on $[0, k_{a}]$; and 
         (b) \(\int_{-\infty}^\infty J_0(\rho x) e^{-x^2} \,\mathrm{d}x\) on~$[-k_{b}, k_{b}]$. And $k_{a}$ and~$k_{b}$ are determined based on Thm.~\ref{thm:truncation_error}.
    }
    \label{fig:Gauss_int_n}
\end{figure}

We validate the proposed quadrature rule by testing its convergence on two benchmark integrals involving the zeroth-order Bessel function of the first kind, i.e., $\int_0^{\infty} \fJ_0( \rho x)  e^{-x} \mrm{d}x$ and $\int_{-\infty}^{\infty} \fJ_0(\rho x) e^{-x^2} \mrm{d}x$. Furthermore, we  compare our results with standard Gauss-Laguerre and Gauss-Hermite quadratures, which are specifically tailored for integrals of these forms.
First, fixing~$\rho = 50$, the absolute error $\mathcal E$ versus the number of quadrature points $n$ is presented in Fig.~\ref{fig:Gauss_int}. Our proposed truncated Gauss–Legendre rule converges significantly faster than Gauss–Laguerre and Gauss–Hermite quadrature, consistent with the findings in \cite{trefethen2022exactness}.
Moreover,  for varying~$\rho \in [0, 10]$, we set the absolute error~$\mathcal E = 10^{-4}$, $10^{-8}$ and $10^{-12}$ (by choosing a suitable $M$ according to Thm.~\ref{thm:truncation_error}), respectively.  We examine the minimal number of quadrature points $n$ required for the proposed quadrature rule to achieve the desired accuracy.
As shown in Fig.~\ref{fig:Gauss_int_n}, only a few tens of points suffice for both integrals, demonstrating the efficiency of the method.

%%%%%%%%%%%%%%%%%%%%%%%%%%%%%%%%%%%%%%%%%%%%%%%%%%%%%%%%%%%%
%%%%%%%%%%%%%%%%%%%%%%%%%%%%%%%%%%%%%%%%%%%%%%%%%%%%%%%%%%%%
%%%%%%%%%%%%%%%%%%%%%%%%%%%%%%%%%%%%%%%%%%%%%%%%%%%%%%%%%%%%
%%%%%%%%%%%%%%%%%%%%%%%%%%%%%%%%%%%%%%%%%%%%%%%%%%%%%%%%%%%%
%%%%%%%%%%%%%%%%%%%%%%%%%%%%%%%%%%%%%%%%%%%%%%%%%%%%%%%%%%%%
%%%%%%%%%%%%%%%%%%%%%%%%%%%%%%%%%%%%%%%%%%%%%%%%%%%%%%%%%%%%
%%%%%%%%%%%%%%%%%%%%%%%%%%%%%%%%%%%%%%%%%%%%%%%%%%%%%%%%%%%%
%\subsection{Efficient Calculation of  Long Range Interaction Energy}
\section{Fast algorithms for long-range interactions}
%%%%%%%%%%%%%%%%%%%%%%%%%%%%%%%%%%%%%%%%%%%%%%%%%%%%%%%%%%%%
%%%%%%%%%%%%%%%%%%%%%%%%%%%%%%%%%%%%%%%%%%%%%%%%%%%%%%%%%%%%
Now we present efficient methods to calculate the long-range interaction energy $U_l$. 
First, based on the Poisson summation formula, we have
%%%%%%%%%%%%%%%%%%%%%%%%%%%%%%%%%%%%%%%%%%%%%%%%%%%%%%%%%%%%
%%%%%%%%%%%%%%%%%%%%%%%%%%%%%%%%%%%%%%%%%%%%%%%%%%%%%%%%%%%%
% \begin{align*}
%  &G_l(\V{r}_i, \V{r}_j)\\
%  =&   \frac{1}{4\pi^2}\int_{\sR^2} \hat{G}(\vk, z_i; z_j)e^{-\frac{k^2}{4\alpha}}\exp{(\mathrm{i} \V{k} \cdot \V{\rho}_i)} \D \vk\\
%  =& \frac{1}{4\pi^2}\int_{\sR^2} \frac{ \exp{(\mathrm{i} \V{k} \cdot (\V{\rho}_i-\V{\rho}_j))}}{2k\epsilon_{\mrm{c}}\left(1-\gamma_{\mrm{u}}\gamma_{\mrm{d}}e^{-2kL_z}\right)} e^{-\frac{k^2}{4\alpha}}\sum_{p=1}^4\Gamma_pe^{-ka_p(z_i;z_j)}\D \vk\\
%  =&\frac{1}{8\pi^2\epsilon_{\mrm{c}}}\int_{\sR^2} \frac{e^{\mathrm{i} \V{k} \cdot \V{\rho}_{ij}}}{k\left(1-\gamma_{\mrm{u}}\gamma_{\mrm{d}}e^{-2kL_z}\right)} e^{-\frac{k^2}{4\alpha}}\sum_{p=1}^4\Gamma_pe^{-ka_p(z_i;z_j)}\D \vk\;,
% \end{align*}
%%%%%%%%%%%%%%%%%%%%%%%%%%%%%%%%%%%%%%%%%%%%%%%%%%%%%%%%%%%%
\begin{equation}\label{eq...text...PoissonSummation}
 \sum_{\vm} G_l(\V{r}_i; \V{r}_j+\vL_m)  =\frac{1}{ L_xL_y}   \sum_{\vk\in\fK^2} \Hat{G}_l(\vk,z_i;z_j)e^{\textrm{i}\vk\cdot \vrho}\;,
\end{equation}
%%%%%%%%%%%%%%%%%%%%%%%%%%%%%%%%%%%%%%%%%%%%%%%%%%%%%%%%%%%%
where \[\fK^2 := \left\{\vk \in \frac{  2\pi}{L_x}\sZ\times \frac{  2\pi}{L_y}\sZ \right\}\;.\]
Substituting $\Hat{G}_l(\vk,z_i;z_j)$ derived in Cor.~\ref{corollary1}, we obtain the following formula for $U_l$,
%%%%%%%%%%%%%%%%%%%%%%%%%%%%%%%%%%%%%%%%%%%%%%%%%%%%%%%%%%%%
% Furthermore, the RHS of Eq.~\eqref{eq...text...PoissonSummation} reads
% %%%%%%%%%%%%%%%%%%%%%%%%%%%%%%%%%%%%%%%%%%%%%%%%%%%%%%%%%%%%
% \begin{align*}
% &\frac{1}{ L_xL_y}   \sum_{\vk\in\fK^2} \Hat{G}_l(\vk,z;z_j)\exp(\textrm{i}\vk\cdot \vrho)\\
% %=&\frac{1}{ L_xL_y}  \left( \sum_{\vk\in\fK^2}  \hat{G}(\vk, z; z_j)e^{-\frac{k^2}{4\alpha}}e^{\mathrm{i} \V{k} \cdot \V{\rho}} \right)\\
% =& \frac{1}{ L_xL_y}   \left[\hat{G}_l(\vzero, z; z_j)+\sum_{\vk\in\fK^2,\vk\neq\vzero}  \hat{G}_l(\vk, z; z_j)e^{\mathrm{i} \V{k} \cdot \V{\rho}}\right]\\
% =&\frac{1}{ L_xL_y}   \left[\hat{G}(\vzero, z; z_j)+\sum_{\vk\in\fK^2,\vk\neq\vzero}  \hat{G}(\vk, z; z_j)e^{-\frac{k^2}{4\alpha}}e^{\mathrm{i} \V{k} \cdot \V{\rho}}\right]\;,
% \end{align*}
% %%%%%%%%%%%%%%%%%%%%%%%%%%%%%%%%%%%%%%%%%%%%%%%%%%%%%%%%%%%%
% therefore, given $\vr_i$, the above expression can be simplified into
% %%%%%%%%%%%%%%%%%%%%%%%%%%%%%%%%%%%%%%%%%%%%%%%%%%%%%%%%%%%%
% \begin{align*}
%  &\frac{1}{ L_xL_y}   \sum_{\vk\in\fK^2,\vk\neq\vzero} \Hat{G}_l(\vk,z;z_j)\exp(\textrm{i}\vk\cdot \vrho)\\
%  =& \frac{1}{ L_xL_y}   \sum_{\vk\in\fK^2,\vk\neq\vzero} \frac{ e^{\mathrm{i} \V{k} \cdot \V{\rho}_{ij}}}{2k\epsilon_{\mrm{c}}\left(1-\gamma_{\mrm{u}}\gamma_{\mrm{d}}e^{-2kL_z}\right)}\left(\sum_{p=1}^4\Gamma_pe^{-ka_p(z_i;z_j)}\right) e^{-\frac{k^2}{4\alpha}}\;,
% \end{align*}
%%%%%%%%%%%%%%%%%%%%%%%%%%%%%%%%%%%%%%%%%%%%%%%%%%%%%%%%%%%%
%%%%%%%%%%%%%%%%%%%%%%%%%%%%%%%%%%%%%%%%%%%%%%%%%%%%%%%%%%%%
%%%%%%%%%%%%%%%%%%%%%%%%%%%%%%%%%%%%%%%%%%%%%%%%%%%%%%%%%%%%
\begin{equation}\label{eq:ul_detailed}
\begin{aligned}
U_l&=\frac{1}{2}\sum_{\vm}\sum_{i,j=1}^N q_iq_j     G_l(\V{r}_i, \V{r}_j+\vL_m)  \\
&=\frac{1}{4\epsilon_{\mrm{c}} L_xL_y}\sum_{i,j=1}^N  q_iq_j \Bigg\{{\Abs{z_i-z_j}} +   \\&~~~~+  \sum_{\vk\in\fK^2,\vk\neq\vzero} \frac{ e^{\mathrm{i} \V{k} \cdot \V{\rho}_{ij}}}{k\left(1-\gamma_{\mrm{u}}\gamma_{\mrm{d}}e^{-2kL_z}\right)}\left(\sum_{p=1}^4\Gamma_pe^{-ka_p(z_i;z_j)}\right) e^{-\frac{k^2}{4\alpha}}\Bigg\}\;.
\end{aligned}\end{equation}
%%%%%%%%%%%%%%%%%%%%%%%%%%%%%%%%%%%%%%%%%%%%%%%%%%%%%%%%%%%%
To reduce the computational cost for Eq.~\eqref{eq:ul_detailed} to $\fO(N)$, we integrate the random batch sampling method with an efficient pairwise summation technique.
The details of our approach will be discussed in the remainder of this section.

\subsection{Part one: $\fO(N)$ pairwise summation method}
%Firstly, we shall deal with the summation over $i$ and $j$, i.e., $\sum_{i,j=1}^N$. 
First, for simplicity, we define 
%%%%%%%%%%%%%%%%%%%%%%%%%%%%%%%%%%%%%%%%%%%%%%%%%%%%%%%%%%%%
% \begin{align*}
%      &\sum_{p=1}^4\Gamma_pe^{-ka_p(z_i;z_j)}\\  
%      =&e^{-k\Abs{z_i-z_j}}+\gamma_{\mrm{d}}e^{-k(z_i+z_j)}\\
%      &+\gamma_{\mrm{u}}\exp(-k(2L_z-z_i-z_j)) +\gamma_{\mrm{u}}\gamma_{\mrm{d}}\exp(-k(2L_z-\Abs{z_i-z_j}))\;.
% \end{align*}
%%%%%%%%%%%%%%%%%%%%%%%%%%%%%%%%%%%%%%%%%%%%%%%%%%%%%%%%%%
auxiliary functions $h_{1:4}(\vk)$ as 
%%%%%%%%%%%%%%%%%%%%%%%%%%%%%%%%%%%%%%%%%%%%%%%%%%%%%%%%%%%%
\begin{align*}
 h_1(\vk)&:= \frac{1}{4\epsilon_{\mrm{c}} L_xL_yk\left(1-\gamma_{\mrm{u}}\gamma_{\mrm{d}}e^{-2kL_z}\right)} \;,\\
  h_2(\vk)&:= \frac{\gamma_{\mrm{d}}}{4\epsilon_{\mrm{c}} L_xL_yk\left(1-\gamma_{\mrm{u}}\gamma_{\mrm{d}}e^{-2kL_z}\right)} \;,\\
   h_3(\vk)&:= \frac{\gamma_{\mrm{u}}e^{-2kL_z}}{4\epsilon_{\mrm{c}} L_xL_yk\left(1-\gamma_{\mrm{u}}\gamma_{\mrm{d}}e^{-2kL_z}\right)}\;,\\
     h_4(\vk)&:= \frac{\gamma_{\mrm{u}}\gamma_{\mrm{d}}e^{-2kL_z}}{4\epsilon_{\mrm{c}} L_xL_yk\left(1-\gamma_{\mrm{u}}\gamma_{\mrm{d}}e^{-2kL_z}\right)}\;,
\end{align*}
%%%%%%%%%%%%%%%%%%%%%%%%%%%%%%%%%%%%%%%%%%%%%%%%%%%%%%%%%%%%
and another set of auxiliary functions $S_{0:4}(\vk)$ as
%which allows us to express the summation over $i$ and $j$ in  $U_l$ as
%%%%%%%%%%%%%%%%%%%%%%%%%%%%%%%%%%%%%%%%%%%%%%%%%%%%%%%%%%%%
\begin{align*} 
 S_0&:=\frac{1}{4\epsilon_{\mrm{c}} L_xL_y}\sum_{i,j=1}^Nq_iq_j {\Abs{z_i-z_j}}\;,\\
S_1(\vk)&:=\sum_{i,j=1}^Nq_iq_j  e^{\mathrm{i} \V{k} \cdot \V{\rho}_{ij}}e^{-k\Abs{z_i-z_j}}\;,\\
S_2(\vk)&:=\sum_{i,j=1}^Nq_iq_j  e^{\mathrm{i} \V{k} \cdot \V{\rho}_{ij}}e^{-k(z_i+z_j)}\;,\\
S_3(\vk)&:=\sum_{i,j=1}^Nq_iq_j  e^{\mathrm{i} \V{k} \cdot \V{\rho}_{ij}}e^{k(z_i+z_j)}\;,\\
S_4(\vk)&:=\sum_{i,j=1}^Nq_iq_j e^{\mathrm{i} \V{k} \cdot \V{\rho}_{ij}}e^{k\Abs{z_i-z_j}}\;.
\end{align*}
%%%%%%%%%%%%%%%%%%%%%%%%%%%%%%%%%%%%%%%%%%%%%%%%%%%%%%%%%%%%
Consequently, the long-range interaction energy $U_l$ can be written as,
%%%%%%%%%%%%%%%%%%%%%%%%%%%%%%%%%%%%%%%%%%%%%%%%%%%%%%%%%%%%
\begin{equation}\label{eq...text...Energy}
U_l =     S_0+ \sum_{\vk\in\fK^2,\vk\neq\vzero} e^{-\frac{k^2}{4\alpha}}\left(\sum_{p=1}^4S_p(\vk)h_p(\vk)\right)\;.
\end{equation}
%%%%%%%%%%%%%%%%%%%%%%%%%%%%%%%%%%%%%%%%%%%%%%%%%%%%%%%%%%%%
%However, prior to this step, we shall  introduce some alignment techniques to further enhance computational efficiency.
%%%%%%%%%%%%%%%%%%%%%%%%%%%%%%%%%%%%%%%%%%%%%%%%%%%%%%%%%%%%
%%%%%%%%%%%%%%%%%%%%%%%%%%%%%%%%%%%%%%%%%%%%%%%%%%%%%%%%%%%%
%%%%%%%%%%%%%%%%%%%%%%%%%%%%%%%%%%%%%%%%%%%%%%%%%%%%%%%%%%%%
%%%%%%%%%%%%%%%%%%%%%%%%%%%%%%%%%%%%%%%%%%%%%%%%%%%%%%%%%%%%
Evaluating  $S_{0:4}(\vk)$  by direct summation would incur  an  $\fO(N^2)$ computational cost.
However, the auxiliary functions $S_2(\vk)$ and $S_3(\vk)$ are \emph{multiplicatively separable}; that is, they can be expressed as
%%%%%%%%%%%%%%%%%%%%%%%%%%%%%%%%%%%%%%%%%%%%%%%%%%%%%%%%%%%%
\begin{align*}
S_2(\vk)&=\sum_{i,j=1}^Nq_iq_j  e^{\mathrm{i} \V{k} \cdot \V{\rho}_{ij}}e^{-k(z_i+z_j)}=\left(\sum_{i=1}^Nq_ie^{\mathrm{i} \V{k} \cdot \V{\rho}_{i}}e^{-kz_i}\right)\left(\sum_{j=1}^Nq_je^{-\mathrm{i} \V{k} \cdot \V{\rho}_{j}}e^{-kz_j}\right)\;,\\
S_3(\vk)&=\sum_{i,j=1}^Nq_iq_j  e^{\mathrm{i} \V{k} \cdot \V{\rho}_{ij}}e^{k(z_i+z_j)}=\left(\sum_{i=1}^Nq_i  e^{\mathrm{i} \V{k} \cdot \V{\rho}_{i}}e^{kz_i}\right)\left(\sum_{j=1}^Nq_j  e^{-\mathrm{i} \V{k} \cdot \V{\rho}_{j}}e^{kz_j}\right)\;,
\end{align*}
%%%%%%%%%%%%%%%%%%%%%%%%%%%%%%%%%%%%%%%%%%%%%%%%%%%%%%%%%%%%
which reduces their computational complexity to $\fO(N)$.
In contrast, $ S_0$,  $S_1(\vk)$ and $S_4(\vk)$  are non-separable due to the presence of the absolute value term $\Abs{z_i - z_j}$, obstructing direct factorization.
Consequently, to eliminate the absolute value operation, we first sort all $N$ particles according to their positions along the $z$-axis, ensuring the order:
\begin{equation*}
   0 < z_1 < z_2 <  \cdots <  z_{N - 1}<z_N < L_z\;.
\end{equation*}
Since for all $i$, $z_i \in (0, L_z)$ is uniformly bounded within the simulation box, strategies such as the bucket sorting~\cite{corwin2004sorting} can be applied, and the computational cost can be of order~$\mathcal{O}(N)$. Especially considering actual particle-based simulations, only local corrections are needed for successive time steps.
Now, once all $N$ particles are sorted, $ S_0$ can be written as
%%%%%%%%%%%%%%%%%%%%%%%%%%%%%%%%%%%%%%%%%%%%%%%%%%%%%%%%%%%%
\begin{align*}
     S_0 
    & = \sum_{i = 1}^N q_i \left( \sum_{j = 1}^{N}{q_j (z_i - z_j)} \right) = 2 \sum_{i = 1}^N q_i \left( z_i \sum_{j = 1}^{i }q_j - \sum_{j = 1}^{i }z_j q_j  \right) \;.
\end{align*}
%%%%%%%%%%%%%%%%%%%%%%%%%%%%%%%%%%%%%%%%%%%%%%%%%%%%%%%%%%%%
To efficiently compute this sum, we further introduce two auxiliary sequences $\left\{T_1(i)\right\}_{i=1}^N$ and $\left\{T_2(i)\right\}_{i=1}^N$, both of which can be iteratively updated as follows:
%%%%%%%%%%%%%%%%%%%%%%%%%%%%%%%%%%%%%%%%%%%%%%%%%%%%%%%%%%%%
\begin{align*}
    &T_1(1)=q_1,~~T_1(i+1)=T_1(i)+q_{i+1}\;,\\
    &T_2(1)=z_1q_1,~~T_2(i+1)=T_2(i)+z_{i+1}q_{i+1}\;.
\end{align*}
%%%%%%%%%%%%%%%%%%%%%%%%%%%%%%%%%%%%%%%%%%%%%%%%%%%%%%%%%%%%
Finally, using these auxiliary sequences, $ S_0$ can be efficiently calculated as
%%%%%%%%%%%%%%%%%%%%%%%%%%%%%%%%%%%%%%%%%%%%%%%%%%%%%%%%%%%%
\[
 S_0= 2 \sum_{i = 1}^N q_i \left( z_iT_1(i)-T_2(i)\right)\;,
\]
%%%%%%%%%%%%%%%%%%%%%%%%%%%%%%%%%%%%%%%%%%%%%%%%%%%%%%%%%%%%
which requires only $\fO(N)$ operations.

A similar strategy~\cite{jiang2021approximating} can be applied for the evaluation of $S_1(\vk)$ and $S_4(\vk)$.
For $S_1(\vk)$, after sorting, we have that 
%%%%%%%%%%%%%%%%%%%%%%%%%%%%%%%%%%%%%%%%%%%%%%%%%%%%%%%%%%%%
\begin{align*}
S_1(\vk)&=\sum_{i,j=1}^Nq_iq_j  e^{\mathrm{i} \V{k} \cdot \V{\rho}_{ij}}e^{-k\Abs{z_i-z_j}}\\
&=-\sum_{i=1}^Nq_i^2+\sum_{i=1}^N \sum_{j=1}^iq_iq_j  e^{\mathrm{i} \V{k} \cdot \V{\rho}_{ij}}e^{-k(z_i-z_j)}+\sum_{i=1}^N \sum_{j=i}^Nq_iq_j  e^{\mathrm{i} \V{k} \cdot \V{\rho}_{ij}}e^{-k(z_j-z_i)}\\
&= -\sum_{i=1}^Nq_i^2+\sum_{i=1}^N q_ie^{\mathrm{i} \V{k} \cdot \V{\rho}_{i}}e^{-kz_i}\sum_{j=1}^iq_je^{-\mathrm{i} \V{k} \cdot \V{\rho}_{j}}e^{kz_j}+\sum_{i=1}^N q_ie^{\mathrm{i} \V{k} \cdot \V{\rho}_{i}}e^{kz_i}\sum_{j=i}^Nq_je^{-\mathrm{i} \V{k} \cdot \V{\rho}_{j}}e^{-kz_j}\;.
\end{align*}
%%%%%%%%%%%%%%%%%%%%%%%%%%%%%%%%%%%%%%%%%%%%%%%%%%%%%%%%%%%%
 We introduce another two auxiliary sequences  $\left\{T_3(i)\right\}_{i=1}^N$ and $\left\{T_4(i)\right\}_{i=1}^N$, and it is noteworthy that $\left\{T_3(i)\right\}_{i=1}^N$ is  constructed through forward iteration, whereas  $\left\{T_4(i)\right\}_{i=1}^N$   is generated through backward iteration:
 %%%%%%%%%%%%%%%%%%%%%%%%%%%%%%%%%%%%%%%%%%%%%%%%%%%%%%%%%%%%
\begin{align*}
    &T_3(1)=q_1e^{-\mathrm{i} \V{k} \cdot \V{\rho}_{1}}e^{kz_1}\;,\\
    &T_3(i+1)=T_3(i)+q_{i+1}e^{-\mathrm{i} \V{k} \cdot \V{\rho}_{i+1}}e^{kz_{i+1}}\;,\\
    &T_4(N)=q_Ne^{-\mathrm{i} \V{k} \cdot \V{\rho}_{N}}e^{-kz_N}\;,\\
    &T_4(i-1)=T_4(i)+q_{i-1}e^{-\mathrm{i} \V{k} \cdot \V{\rho}_{i-1}}e^{-kz_{i-1}}\;,
\end{align*}
%%%%%%%%%%%%%%%%%%%%%%%%%%%%%%%%%%%%%%%%%%%%%%%%%%%%%%%%%%%%
and
$S_1(\vk)$ can be efficiently calculated as
%%%%%%%%%%%%%%%%%%%%%%%%%%%%%%%%%%%%%%%%%%%%%%%%%%%%%%%%%%%%
\[
S_1(\vk)= -\sum_{i=1}^Nq_i^2+ \sum_{i=1}^N q_ie^{\mathrm{i} \V{k} \cdot \V{\rho}_{i}}\left[e^{-kz_i}T_3(i)+e^{kz_i}T_4(i)\right]\;.
\]
%%%%%%%%%%%%%%%%%%%%%%%%%%%%%%%%%%%%%%%%%%%%%%%%%%%%%%%%%%%%
We observe that the same  recurrence scheme can be applied to the calculation of $S_4(\vk)$, and the details are omitted for brevity.

In summary, for any given $\vk$ in Eq.~\eqref{eq...text...Energy}, it now only takes $\fO(N)$ cost to evaluate $ S_0$ and $\sum_{p=1}^4S_p(\vk)h_p(\vk)$.
We proceed to introduce the random batch method to efficiently perform the summation over $\vk$. 
\subsection{Part two: random batch importance sampling in $k$-space}
%%%%%%%%%%%%%%%%%%%%%%%%%%%%%%%%%%%%%%%%%%%%%%%%%%%%%%%%%%%%
%%%%%%%%%%%%%%%%%%%%%%%%%%%%%%%%%%%%%%%%%%%%%%%%%%%%%%%%%%%%
%%%%%%%%%%%%%%%%%%%%%%%%%%%%%%%%%%%%%%%%%%%%%%%%%%%%%%%%%%%%
In this section, we introduce the random batch method~\cite{jin2021random,jin2020random} for efficient summation over $\vk$, so that the overall computational complexity for $U_l$ can be reduced to $\mathcal{O}(N)$. 
The idea of random mini-batch sampling originates from optimization and machine learning, such as the well-known stochastic gradient descent (SGD) method. S. Jin et al. first proposed the random batch method (RBM) for interacting particle systems~\cite{jin2020random}.
This approach was subsequently extended to sampling in $k$-space using importance sampling for variance reduction. 
The resulting random batch Ewald method (RBE)~\cite{jin2021random, liang2021randomlist, liang2022superscalability, liang2022random, liang2022improved, liang2023random, liang2024energy, gan2025random} has since been successfully applied to various ensembles and systems, demonstrating remarkable superscalability in large-scale simulations.

Here, to further accelerate the computation for the long-range interaction energy $U_l$, we first define $\phi(\vk)$ as,
\begin{equation}\label{eq...text...phi}
    \phi(\vk):=\sum_{p=1}^4S_p(\vk)h_p(\vk)\;,
\end{equation}
then Eq.~\eqref{eq...text...Energy} becomes
%%%%%%%%%%%%%%%%%%%%%%%%%%%%%%%%%%%%%%%%%%%%%%%%%%%%%%%%%%%%
\begin{equation} \label{eq...text...EnergywithPhi}
U_l =     S_0+ \sum_{\vk\in\fK^2,\vk\neq\vzero} \phi(\vk)e^{-\frac{k^2}{4\alpha}} \;.
\end{equation}
We define 
\begin{equation}\label{eq::U_l0_Ulk}
    U_l^{\vzero} := S_0,\quad\text{and}\quad U_l^{\vk\neq \vzero} :=\sum_{\vk\in\fK^2,\vk\neq\vzero}\phi(\vk) e^{-\frac{k^2}{4\alpha}} \;,
\end{equation}
then $U_l$ can be written as $U_l= U_l^{\vzero}+U_l^{\vk\neq \vzero}\;.$
Since 
$ U_l^{\vzero}$
can be computed with $\fO(N)$ cost, we focus on the efficient evaluation of $U_l^{\vk\neq \vzero}$ via random batch method. 

By understanding the lattice summation over $\vk$ as calculating an expectation value under the probability measure $\mathcal P(\vk)\propto e^{\frac{-k^2}{4\alpha}}$ defined on the lattice $\mathcal{K}^2$, we can build a stochastic approximation for $U_l^{\vk\neq \vzero}$ via importance sampling,
\begin{equation}\label{eq::RBapp}
 U_l^{\vk\neq \vzero} \approx U_{l,\ast}^{\vk\neq \vzero}  :=   \frac{H}{P}\sum_{\eta=1}^P \phi(\vk_{\eta})=  \frac{H}{P}\sum_{\eta=1}^P \sum_{q=1}^4S_q(\vk_{\eta})h_q(\vk_{\eta})\;,
\end{equation} 
where $\{\bm{k}_{\eta}\}_{\eta=1}^P$ are a mini-batch of frequencies sampled from $\mathcal P(\vk)$ using the Metropolis algorithm~\cite{metropolis1953equation}, where~$P$ is the batch size,
%Based on the above expression, we find a Gaussian decay factor in $U_l^{\vk\neq \vzero}$ explicitly, which can be normalized for the  purpose of importance sampling, so that we take
% \begin{equation}\label{eq::hk}
% 	g(\bm{k}) := \frac{1}{H} e^{-\frac{k^2}{4\alpha}}\;, %\quad H := \sum_{\bm{k}\neq\bm{0}}e^{-k^2/(4\alpha^2)},
% \end{equation}
and \[H:= \sum_{\vk\in\fK^2,\vk\neq\vzero}e^{-\frac{k^2}{4\alpha}}\] the normalization constant for $\mathcal P$.
% By application of the Poisson summation formula, $H$ can be computed as follows
% %%%%%%%%%%%%%%%%%%%%%%%%%%%%%%%%%%%%%%%%%%%%%%%%%%%%%%%%%%%%
% \begin{align*}
%   H+1&=  \sum_{\vk\in\fK^2}e^{-\frac{k^2}{4\alpha}}\;,\\
%   \sum_{\vk\in\fK^2}e^{-\frac{k^2}{4\alpha}}&= \frac{\alpha L_xL_y}{\pi}\sum_{m_x,m_y\in\mathbb{Z}}\exp\left({-\alpha\left(m_x^2L_x^2+m_y^2L_y^2\right)}\right)\;,
% \end{align*}
% %%%%%%%%%%%%%%%%%%%%%%%%%%%%%%%%%%%%%%%%%%%%%%%%%%%%%%%%%%%%
% hence
% %%%%%%%%%%%%%%%%%%%%%%%%%%%%%%%%%%%%%%%%%%%%%%%%%%%%%%%%%%%%
% \begin{equation}\label{eq::Happ}
% 	H=-1+\frac{\alpha L_xL_y}{\pi}\sum_{m_x,m_y\in\mathbb{Z}}\exp\left({-\alpha\left(m_x^2L_x^2+m_y^2L_y^2\right)}\right)\;,
% \end{equation}
% %%%%%%%%%%%%%%%%%%%%%%%%%%%%%%%%%%%%%%%%%%%%%%%%%%%%%%%%%%%%
%\xz{I removed the statement about the computation of $H$ and the truncation, since if we use the Metropolis algorithm, we don't need to compute $H$ explicitly.}
%%%%%%%%%%%%%%%%%%%%%%%%%%%%%%%%%%%%%%%%%%%%%%%%%%%%%%%%%%%%
Then the corresponding stochastic estimator of forces in Fourier space can also be obtained as
\begin{equation}
	 \bm{F}_{l,i}\approx\bm{F}_{l,i}^*=-  \grad_{\V{r}_i} U_l^{\vzero}- \grad_{\V{r}_i} U_{l,\ast}^{\vk\neq \vzero}\;.
 \end{equation}

 In what follows, we present a theoretical analysis of the random batch approximation, emphasizing its convergence rate. We demonstrate that the batch size $P$ is independent of the particle number $N$, which ensures that the summation over $\vk$ does not affect the overall linear scaling.
Let us start with considering the fluctuations, i.e., the stochastic error introduced by the importance sampling at each time step. The fluctuations of the force acting on the $i$th particle due to random batch approximation, denoted as $\V{\chi}_{i}$, are defined as follows 
\begin{equation}\label{eq::xichi}
%	\Xi :=U_{l,\ast}^{\vk\neq \vzero}-U_l^{\vk\neq \vzero},\quad\text{and}\quad
    \V{\chi}_{i} :=\vF_{l,i}^*-\vF_{l,i}=\grad_{\V{r}_i} U_{l,\ast}^{\vk\neq \vzero}-\grad_{\V{r}_i} U_{l}^{\vk\neq \vzero}\;.
\end{equation}
%%%%%%%%%%%%%%%%%%%%%%%%%%%%%%%%%%%%%%%%%%%%%%%%%%%%%%%%%%%%
The unbiasedness of random batch approximation can be obtained straightforwardly, and is summarized in the following proposition.
%%%%%%%%%%%%%%%%%%%%%%%%%%%%%%%%%%%%%%%%%%%%%%%%%%%%%%%%%%%%
\begin{proposition}\label{prop:unbaised}
	  $\vF_{l,i}^*$ is an unbiased estimator, i.e.,  $\mathbb{E} \V{\chi}_{i} = \bm 0$, and its variance  can be expressed as
% \begin{align*}
% 		\mathbb{E}\Xi^2&=\frac{H}{P}\sum_{\bm{k}_1\in\fK^2, \bm{k}_1\neq \vzero}\exp\left(-\frac{\Norm{\vk_1}^2}{4\alpha}\right)\Norm{\phi(\vk_1)-\frac{1}{H} \sum_{\bm{k}_2\in\fK^2, \bm{k}_2\neq \vzero}\phi(\vk_2)\exp\left(-\frac{\Norm{\vk_2}^2}{4\alpha}\right)}^2\;,
%         % \\
%         % \Bigg|\sum_{q=1}^4S_q(\vk_{1})h_q(\vk_{1})\\
%         % &~~~~~~~~~~~~~~~~~~~~~~~~~~~~~~~~~-\frac{1}{H} \sum_{\bm{k}_2\neq \vzero}\exp\left(-\frac{\Norm{\vk_2}^2}{4\alpha}\right)\sum_{q=1}^4S_q(\vk_{2})h_q(\vk_{2}) \Bigg|^2\;,
% \end{align*}
% 	and
\begin{align*}
		\mathbb{E}\Norm{\V{\chi}_{i}}^2=&\frac{H}{P}\sum_{\bm{k}\in\fK^2, \bm{k}\neq \vzero}e^{-\frac{\Norm{\vk}^2}{4\alpha}} 
        \Bigg|\Bigg|\nabla_{\bm{r}_i}\phi(\vk)
        -\frac{1}{H} \sum_{\bm{k}'\in\fK^2, \bm{k}'\neq \vzero}\nabla_{\bm{r}_i}\phi(\vk')e^{-\frac{\Norm{\vk'}^2}{4\alpha}}\Bigg|\Bigg|^2\;.
\end{align*}
\end{proposition}
%%%%%%%%%%%%%%%%%%%%%%%%%%%%%%%%%%%%%%%%%%%%%%%%%%%%%%%%%%%%
Furthermore, under the Debye-H$\ddot{\text{u}}$ckel  (DH) approximation~\cite{levin2002electrostatic}, we have the following~Lemma~\ref{lem::upper_bound_phiRB}, whose proof can be found in the Supplementary Materials.%\rev{ZG: Appendix F in~\cite{gan2024fast}.}
\begin{lemma} \label{lem::upper_bound_phiRB}
  Under the assumption of the DH theory, given function of the form 
    \begin{equation}
       \mathscr{G}(\V{r}_i)=\sum_{j\neq i}q _i q_{j}e^{\mathrm{i} \V{k} \cdot \V{\rho}_{ij}}f(z_{ij}),
    \end{equation}
 where $\Abs{f(z_{ij})}$ is bounded by  a constant $C_f$ independent of $z_{ij}$, then the function $ \mathscr{G}(\V{r}_i)$ is bounded above by
    \begin{equation}
         \mathscr{G}(\V{r}_i)\leq q_i  C_f\lambda_{\mrm{D}}^2\;,
    \end{equation}
    where~$\lambda_{\mrm{D}}$  represents the Debye length.
	% Under the assumption of the DH theory, $|\phi(\bm{k})|$ and $\Norm{\nabla_{\bm{r}_i}\phi(\bm{k})}$ have their respective upper bounds
	% \begin{equation}
	% 	\Norm{ \widetilde{\varphi}^{\emph{RB}}(\bm{k})}\leq\frac{2\sqrt{\pi}\lambda_D^2 Q}{L_xL_y k}\left(\sqrt{\pi}+\frac{\alpha \varepsilon}{k}\right),
	% 	\quad
	% 	\Norm{\nabla_{\bm{r}_{i}} \widetilde{\varphi}^{\emph{RB}}(\bm{k})}\leq  \frac{\pi \lambda_D^2 q_{i}^2}{L_xL_y}\left[3+\frac{\alpha}{\sqrt{\pi}}\left(1+\frac{2\sqrt{2}\varepsilon}{k}\right)\right], 
	% \end{equation}
	% where $\lambda_{D}$ represents the Debye length
\end{lemma}
%%%%%%%%%%%%%%%%%%%%%%%%%%%%%%%%%%%%%%%%%%%%%%%%%%%%%%%%%%%%
%%%%%%%%%%%%%%%%%%%%%%%%%%%%%%%%%%%%%%%%%%%%%%%%%%%%%%%%%%%%
Finally, we introduce the following Theorem~\ref{thm:unbaised} for the boundedness and convergence in the fluctuations originated from the random batch approximation.
%%%%%%%%%%%%%%%%%%%%%%%%%%%%%%%%%%%%%%%%%%%%%%%%%%%%%%%%%%%%
\begin{theorem}\label{thm:unbaised}
	Under the assumption of the DH theory,   the variance  of the estimators of   forces has a closed upper bound
	\begin{equation}
		  \mathbb{E}\Norm{\V{\chi}_{i}}^2\leq\frac{ H^2q_i^2\lambda_{\mrm{D}}^4}{2P\epsilon_{\mrm{c}}^2 L_x^2L_y^2}\left( \frac{1+\Abs{\gamma_{\mrm{u}}}+\Abs{\gamma_{\mrm{d}}}+\Abs{\gamma_{\mrm{u}}\gamma_{\mrm{d}}}}{1-\max\{0,\gamma_{\mrm{u}}\gamma_{\mrm{d}}\}}\right)^2\;,
	\end{equation}
     where~$\lambda_{\mrm{D}}$  represents the Debye length.
\end{theorem}
%%%%%%%%%%%%%%%%%%%%%%%%%%%%%%%%%%%%%%%%%%%%%%%%%%%%%%%%%%%%
\begin{proof}
We observe that 
%%%%%%%%%%%%%%%%%%%%%%%%%%%%%%%%%%%%%%%%%%%%%%%%%%%%%%%%%%%%
\begin{align*}
 \nabla_{\bm{r}_i}\phi(\vk)&=   \sum_{q=1}^4 \nabla_{\bm{r}_i}\left\{S_q(\vk)\right\}h_q(\vk)\;,
\end{align*}
%%%%%%%%%%%%%%%%%%%%%%%%%%%%%%%%%%%%%%%%%%%%%%%%%%%%%%%%%%%%
thus we have,
% %%%%%%%%%%%%%%%%%%%%%%%%%%%%%%%%%%%%%%%%%%%%%%%%%%%%%%%%%%%%
% \begin{align*}
%  \nabla_{\bm{r}_i} S_q(\vk)&= \begin{pmatrix} \nabla_{\vrho_i}S_q(\vk)\\
% \partial_{z_i}S_q(\vk) \end{pmatrix}   = \begin{pmatrix} \mathrm{i} \V{k}S_q(\vk)\\
% \pm k S_q(\vk) \end{pmatrix} =\begin{pmatrix} \mathrm{i} \V{k} \\
% \pm k  \end{pmatrix}S_q(\vk)\;,
% \end{align*}
% %%%%%%%%%%%%%%%%%%%%%%%%%%%%%%%%%%%%%%%%%%%%%%%%%%%%%%%%%%%%
%%%%%%%%%%%%%%%%%%%%%%%%%%%%%%%%%%%%%%%%%%%%%%%%%%%%%%%%%%%%
\begin{align*}
 \nabla_{\bm{r}_i} S_1(\vk)&=\sum_{j,j\neq i}q_iq_j   \begin{pmatrix} \mathrm{i} \V{k} \\
 \pm k  \end{pmatrix}e^{\mathrm{i} \V{k} \cdot \V{\rho}_{ij}}e^{-k\Abs{z_i-z_j}}\;,\\
 \nabla_{\bm{r}_i} S_2(\vk)&=\sum_{j,j\neq i}q_iq_j   \begin{pmatrix} \mathrm{i} \V{k} \\
 -k  \end{pmatrix}e^{\mathrm{i} \V{k} \cdot \V{\rho}_{ij}}e^{-k(z_i+z_j)}\;,\\
  \nabla_{\bm{r}_i} S_3(\vk)&=\sum_{j,j\neq i}q_iq_j   \begin{pmatrix} \mathrm{i} \V{k} \\
 k  \end{pmatrix}e^{\mathrm{i} \V{k} \cdot \V{\rho}_{ij}}e^{k(z_i+z_j)}\;,\\
 \nabla_{\bm{r}_i} S_4(\vk)&=\sum_{j,j\neq i}q_iq_j   \begin{pmatrix} \mathrm{i} \V{k} \\
 \pm k  \end{pmatrix}e^{\mathrm{i} \V{k} \cdot \V{\rho}_{ij}}e^{k\Abs{z_i-z_j}}\;,\\
\end{align*}
%%%%%%%%%%%%%%%%%%%%%%%%%%%%%%%%%%%%%%%%%%%%%%%%%%%%%%%%%%%%
and  the choice of $+$ or $-$ depends on the choice of particle. Consequently, $\grad_{\V{r}_i} \phi(\V{k})$ takes the form
%%%%%%%%%%%%%%%%%%%%%%%%%%%%%%%%%%%%%%%%%%%%%%%%%%%%%%%%%%%%
\begin{align*}
  \grad_{\V{r}_i} \phi(\V{k})  &=\sum_{j,j\neq i}q_iq_j  e^{\mathrm{i} \V{k} \cdot \V{\rho}_{ij}} \begin{pmatrix} \mathrm{i} \V{k} \\
 \pm k  \end{pmatrix}g(z_{ij})\;,
\end{align*}
%%%%%%%%%%%%%%%%%%%%%%%%%%%%%%%%%%%%%%%%%%%%%%%%%%%%%%%%%%%%
where the estimates on $g(z_{ij})$ read
\begin{align*}
  \Abs{g(z_{ij})  }&\leq  \frac{1+\Abs{\gamma_{\mrm{u}}}+\Abs{\gamma_{\mrm{d}}}+\Abs{\gamma_{\mrm{u}}\gamma_{\mrm{d}}}}{4\epsilon_{\mrm{c}} L_xL_y k\left(1-\max\{0,\gamma_{\mrm{u}}\gamma_{\mrm{d}}\}\right)}\;,
\end{align*}
On the basis of Lemma \ref{lem::upper_bound_phiRB}, the upper bound of~$\Norm{\grad_{\V{r}_i} \phi(\V{k})}$ can be estimated as follows 
%%%%%%%%%%%%%%%%%%%%%%%%%%%%%%%%%%%%%%%%%%%%%%%%%%%%%%%%%%%%
\begin{equation}
  \Norm{\nabla_{\bm{r}_i}\phi(\vk)}\leq \frac{q_i\lambda_{\mrm{D}}^2}{2\epsilon_{\mrm{c}} L_xL_y}    \frac{1+\Abs{\gamma_{\mrm{u}}}+\Abs{\gamma_{\mrm{d}}}+\Abs{\gamma_{\mrm{u}}\gamma_{\mrm{d}}}}{1-\max\{0,\gamma_{\mrm{u}}\gamma_{\mrm{d}}\}}\;.
\end{equation}
%%%%%%%%%%%%%%%%%%%%%%%%%%%%%%%%%%%%%%%%%%%%%%%%%%%%%%%%%%%%
Finally, the variance of the force on the $i$-th particle is bounded above by
%%%%%%%%%%%%%%%%%%%%%%%%%%%%%%%%%%%%%%%%%%%%%%%%%%%%%%%%%%%%
\begin{align*}
   	\mathbb{E}\Norm{\V{\chi}_{i}}^2&\leq   \frac{2H}{P}\sum_{\bm{k}\in\fK^2, \bm{k}\neq \vzero}e^{-\frac{\Norm{\vk}^2}{4\alpha}}\Norm{\nabla_{\bm{r}_i}\phi(\vk)}^2\\
    &\leq  \frac{2H^2}{P}\left(\frac{q_i\lambda_{\mrm{D}}^2}{2\epsilon_{\mrm{c}} L_xL_y}    \frac{1+\Abs{\gamma_{\mrm{u}}}+\Abs{\gamma_{\mrm{d}}}+\Abs{\gamma_{\mrm{u}}\gamma_{\mrm{d}}}}{1-\max\{0,\gamma_{\mrm{u}}\gamma_{\mrm{d}}\}}\right)^2\\
    &= \frac{ H^2q_i^2\lambda_{\mrm{D}}^4}{2P\epsilon_{\mrm{c}}^2 L_x^2L_y^2}\left( \frac{1+\Abs{\gamma_{\mrm{u}}}+\Abs{\gamma_{\mrm{d}}}+\Abs{\gamma_{\mrm{u}}\gamma_{\mrm{d}}}}{1-\max\{0,\gamma_{\mrm{u}}\gamma_{\mrm{d}}\}}\right)^2\;.
\end{align*}
\end{proof}
We note that $H\sim L_xL_y$ and $\mathbb{E}\Norm{\V{\chi}_{i}}^2=\fO\left(\frac{1}{P}\right)$, which is \emph{independent} of the particle number $N$.
%As a result, Theorem~\ref{thm:unbaised} guarantees that only $\fO(1)$ samples are needed. 
%This is crucial for its practical usage in MD simulations, where the dynamical evolution typically relies on force calculations rather than energy. 
Next, we demonstrate that such a random batch approximation for forces can capture the finite time dynamics of the Langevin thermostat, i.e., Eq.~\eqref{eq:NewtonEquation}.
%some analysis for the strong convergence of the random batch MD, which further supports this observation. The following result indicates 
Specifically, consider a numerical integration of Eq.~\eqref{eq:NewtonEquation} with $\Delta t$ the discretized time step, $\bm{r}_i$, $m_i$ and $\bm{p}_i$ the position, mass, and momentum of the $i$th particle, respectively. We have the following strong convergence result.
\begin{theorem}(Strong convergence)\label{thm:strong_conv}
Let $(\bm{r}_i, \bm{v}_i)$ be the solutions to
\begin{equation*}
\begin{split}
&\D\bm{r}_i=\bm{v}_i\,\D t\;,\\
&m_i \D\bm{v}_i=\left[\bm{F}_i-\gamma \bm{v}_i\right]\,\D t+\sqrt{\frac{2\gamma}{\beta}}\D \bm{W}_i\;,
\end{split}
\end{equation*}
where $\{\bm{W}_i\}_{i=1}^N$ are i.i.d. Wiener processes. Let $(\widetilde{\bm{r}}_i, \widetilde{\bm{v}}_i)$ be the solutions to
\begin{equation*}
\begin{split}
&\D \widetilde{\bm{r}}_i=\widetilde{\bm{v}}_i\,\D t\;,\\
&m_i \D\widetilde{\bm{v}}_i=\left[\bm{F}_i +\bm{\chi}_i-\gamma \widetilde{\bm{v}}_i\right]\,\D t+\sqrt{\frac{2\gamma}{\beta}}\D \bm{W}_i\;,
\end{split}
\end{equation*}
with the same initial values as $(\bm{r}_i, \bm{v}_i)$. Suppose that the masses $m_i$'s are bounded uniformly from above and below.
If the forces $\bm{F}_i$ are bounded and Lipschitz and $\Exp \bm{\chi}_i=0$, then for any $T>0$, there exists $C(T)>0$ such that
\[
\Exp\left[\frac{1}{N}\sum_i (\Norm{\bm{r}_i-\widetilde{\bm{r}}_i}^2
+\Norm{\bm{v}_i-\widetilde{\bm{v}}_i}^2) \right] \leq C(T)\sqrt{\Lambda \Delta t}\;,
\]
where $\Lambda$ is an upper bound for $\max_i  \Exp\Norm{\bm{\chi}_i}^2$ and is independent of $N$.
\end{theorem}
Similar proofs can be found in~\cite{li2019stochastic,jin2021convergence,ye2024error}, and thus, we omit the proof here.
It is important to note, however, that in practical molecular dynamics (MD) simulations, both electrostatic and Lennard-Jones interactions can violate the required Lipschitz continuity and boundedness conditions as $r \to 0$. 
%Additionally, one may concern that the errors introduced
%by the SOE approximation and Ewald decomposition might disrupt the convergence of the method. 
 Consequently, rigorously justifying convergence with singular interactions is challenging and remains an open question. Nonetheless, we contend that Thm.~\ref{thm:strong_conv} should still hold in practice for two reasons: (1) the Lennard-Jones (LJ) potential, which is employed in MD simulations to represent the finite size of particles, provides strong short-range repulsion, preventing particles from approaching the singularity; and (2) the substantial variance reduction achieved through the importance sampling technique mitigates stability issues. 
 Numerical results presented in the next section will further validate the effectiveness of the random batch method in capturing finite-time structure and dynamic properties, consistent with the conclusions of Thm.~\ref{thm:strong_conv}.

\section{Numerical results}

\label{sec:result}

In this section, we evaluate the accuracy and efficiency of the proposed QEM method using representative charged systems under confinement. 
These examples highlight characteristic multiscale phenomena, including dielectric boundary effects, anisotropic diffusion, and electrical double layer (EDL) structures under charge-asymmetric conditions. 
Simulations were performed on a single core of an AMD Ryzen Threadripper PRO 3995WX @ 2.2 GHz workstation. 
The software implementation~\cite{QuasiEwald} relies on Julia~\cite{Julia-2017} and \emph{CellListMap.jl}~\cite{celllistmap}, with visualizations generated via \emph{Makie.jl}~\cite{DanischKrumbiegel2021}.

\subsection{Convergence validation on static configurations}

We first validate the convergence of our method by calculating electrostatic energy and force fields for static particle configurations under quasi-2D confinement.
%Note that here the random batch method is not applied.
Figure~\ref{fig:E_xyz} compares the electrostatic fields in the $x$ and $z$ directions for a pair of oppositely charged particles under different dielectric confinements against reference values from the Image Charge Method (ICM).
The simulation box is $(100, 100, 50)$. 
A charge $q_1 = +1$ is fixed at $(50, 50, 1)$, while a second charge $q_2 = -1$ moves along either the $x$- or $z$-axis.
For QEM, parameters are $\alpha = 1.0$ and $n = 50$. 
For the ICM benchmark, we use Eq.~\eqref{eq:ICM} with $(200, 200)$ periodic images in $xy$ and $40$ reflections in $z$ (totaling $16,000$ image charges per source) to guarantee high accuracy.
The results in Fig.~\ref{fig:E_xyz} show excellent agreement between QEM and the ICM benchmarks.
Specifically, Fig.~\ref{fig:E_xyz}(a) reveals that confinement by conductor-like materials ($\gamma = -0.95$) weakens the Coulomb attraction between opposite charges, whereas insulator-like confinement ($\gamma = 0.95$) enhances it. 
Furthermore, Fig.~\ref{fig:E_xyz}(b) shows that the dielectric boundary exerts an attraction or repulsion on particles near the wall, arising from self-interaction with image charges.
These results confirm that QEM accurately captures dielectric boundary effects under doubly periodic conditions.

\begin{figure}[htbp!]
    \centering
    \includegraphics[width = \linewidth]{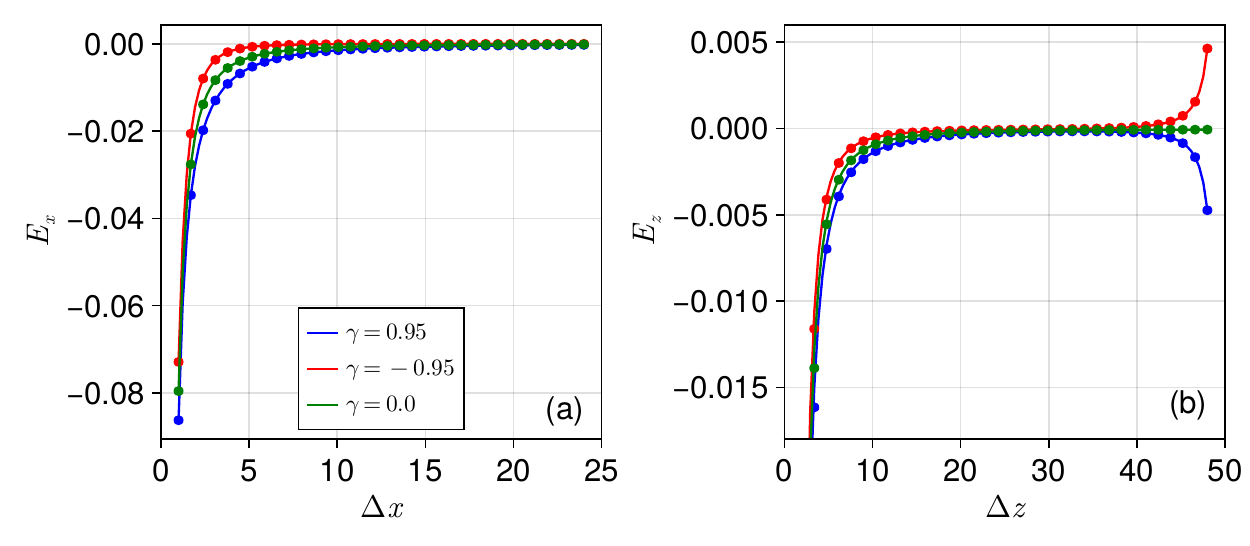}
    \caption{
        Electric fields in (a) $x$ and (b) $z$ between a pair of oppositely charged particles, which are confined in a box with size of $(100, 100, 50)$.
        The positive particle is fixed at $(50, 50, 1)$, while the negative charge is moving along (a) $(50 + \Delta x, 50, 1)$ and (b) $(50, 50, \Delta z)$, respectively.
        The solid lines are results obtained using QEM and  dots represent benchmark values by ICM with direct lattice sum of 16,000 image charges per particle.
    }
    \label{fig:E_xyz}
\end{figure}

We further validate QEM convergence by calculating the electrostatic interaction energy for a system of $100$ randomly distributed particles, comparing against the standard Ewald2D method combined with ICM.
Figure~\ref{fig:Error_E} plots the relative energy error as a function of the reflection factor~$\gamma$ and the error control parameter~$E$.
Here, $E$ is chosen according to Theorem~\ref{thm:truncation_error} such that the truncation errors satisfy $\Norm{\Delta \mathrm{I}_b(M)}\leq 10^{-E}$ and $\Norm{\Delta \mathrm{II}(M)}\leq 10^{-E}$.
The results confirm that the relative error decreases as $E$ increases, consistent with our theoretical analysis.

\begin{figure}[htbp!]
    \centering
    \includegraphics[width = 0.625\linewidth]{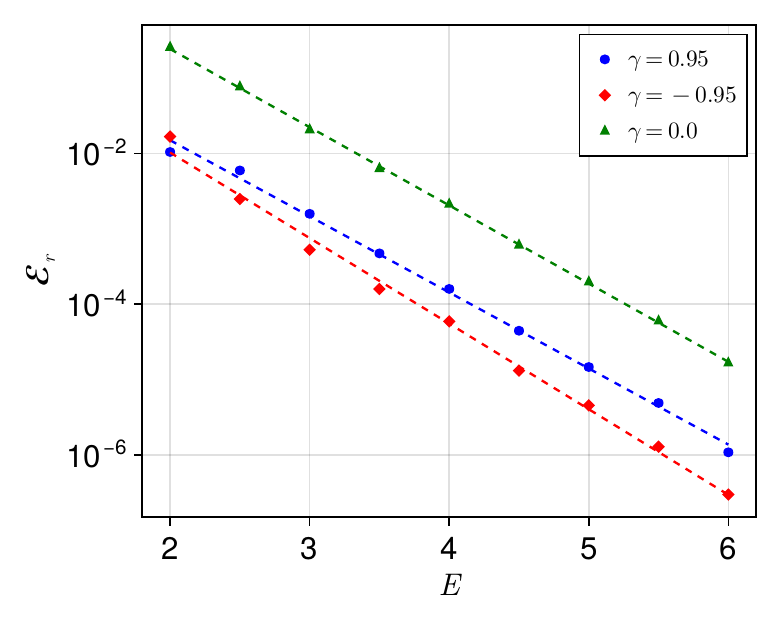}
    \caption{
        Relative error of the electrostatic interaction energy of $100$ charged particles randomly distributed in a box with size of~$(100, 100, 50)$, as a function of the error control parameter~$E$ obtained from our truncation error analysis Thm.~\ref{thm:truncation_error}. 
        The reflection factor~$\gamma$ is chosen as $\pm 0.95$ or $0$.
    }
    \label{fig:Error_E}
\end{figure}

\subsection{Applications in MD simulations}

To demonstrate the practical utility of QEM for MD simulations, we study prototypical ionic systems confined between two dielectric interfaces.
The simulation box has dimensions $(100 \tau_0, 100\tau_0, L_z)$ and contains 436 charged particles in a continuum solvent with Bjerrum length $\ell_{\mathrm B} = e_0^2/(4 \pi \eps_2 k_B T) = 3.5 \tau_0$.
Ions are modeled as soft spheres of diameter $\tau_0$ with excluded-volume interactions given by a purely repulsive shifted-truncated Lennard-Jones (LJ) potential with energy scale $\eps_{LJ} = k_B T$.
Confinement is enforced by purely repulsive shifted-truncated LJ walls ($\eps_\text{ion-wall} = k_B T$; $\tau_\text{ion-wall} = 0.5 \tau_0$) at $z = 0$ and $z = L_z$.
We use a time step of $0.001 t_0$, where $t_0 = \tau_0 m_0 / k_B T$ is the time unit (with ion mass $m_0 = 1$), and control temperature via a Nosé–Hoover thermostat.
QEM parameters are set to $E = 4$, $n = 40$, and batch size $P = 30$.
Systems are equilibrated for $2 \times 10^6$ steps, followed by a production run of $3 \times 10^6$ steps.
Sampling is performed every 100 steps, yielding $3 \times 10^4$ independent configurations for statistical analysis.

\subsubsection{MD simulations of symmetric systems}

We first investigate symmetric systems with identical top and bottom substrate permittivities $\gamma = 0.95$ or $-0.95$ and a channel height $L_z = 50 \tau_0$.
The system comprises 218 cations and 218 anions with symmetric charges $\pm e_0$.
To validate our QEM results, we compare the ion density profiles with benchmark data from the Harmonic Surface Mapping Algorithm (HSMA)~\cite{liang2020harmonic}.
Figure~\ref{fig:MD} shows excellent agreement between QEM and HSMA.
Consistent with the static analysis in Fig.~\ref{fig:E_xyz}(b), we observe ion accumulation near conductor-like interfaces and depletion near insulator-like interfaces.
We also compute the mean square displacements (MSD) in the $xy$ plane and $z$ direction.
As shown in Fig.~\ref{fig:msd}(a), $\mrm{MSD}_{xy}$ exhibits standard bulk-like behavior: a ballistic regime ($\sim t^2$) for $t < 1$ and a diffusive regime ($\sim t$) for $t > 1$.
In contrast, Fig.~\ref{fig:msd}(b) shows that $\mrm{MSD}_{z}$ saturates at large timescales.
This anisotropic diffusion is a direct consequence of confinement.

\begin{figure}[htbp!]
    \centering
    \includegraphics[width = 0.8 \linewidth]{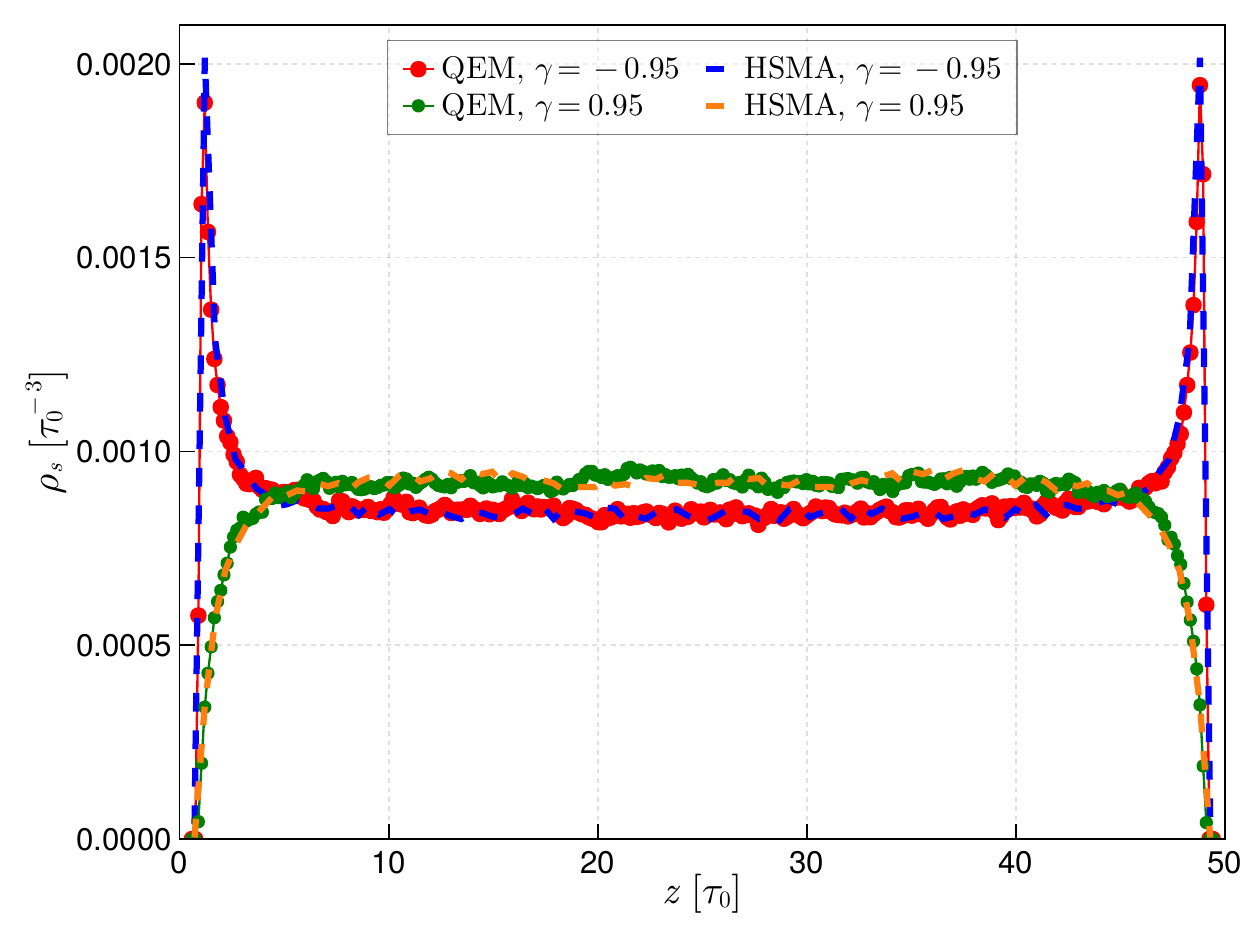}
    \caption{
        Equilibrium distributions of ion densities in~$z$ for symmetric electrolytes. The simulation box contains 218 cations and 218 anions, with dielectric confinements $\gamma_u = \gamma_d = \pm 0.95$ at~$z = 0$ and~$50\tau_0$, respectively. The solid lines indicate results obtained by QEM, and dashed lines are benchmark values using the HSMA method.
    }
    \label{fig:MD}
\end{figure}

\begin{figure}[htbp!]
    \centering
    \includegraphics[width = 1.0 \linewidth]{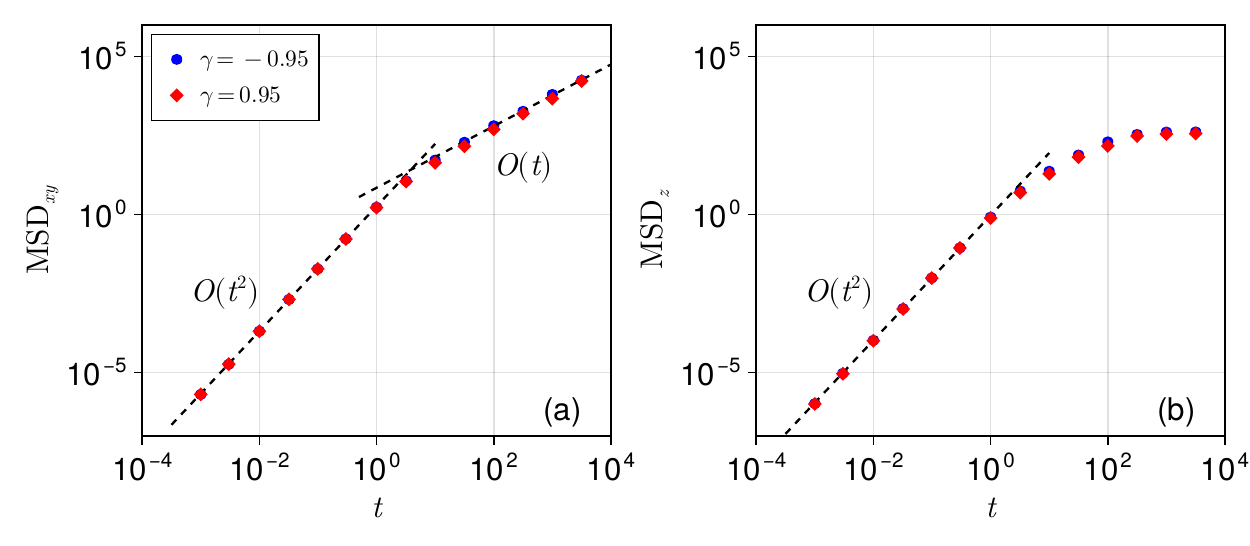}
    \caption{
        The mean square displacements (MSD) of the confined particles in (a)~$xy$ directions and (b)~$z$ direction, respectively, as a function of timescale $t$. Two reflection factor values, $\gamma= 0.95$ and $-0.95$, are considered.
    }
    \label{fig:msd}
\end{figure}

    The CPU time cost for QEM to evaluate the interactions between~$N$ confined particles is shown in Fig.~\ref{fig:timecost}, where the cost of the Ewald2D method is also plotted as a reference.
We fix the thickness of the system in $z$ to be~$50 \tau_0$, and vary the number of charged particles from~$10^2$ to~$10^5$ (but keep the ionic density to be fixed).
Since we choose $E=4$ in QEM, for a fair comparison, the Ewald2D parameters are also chosen to guarantee 4-digits accuracy.
The CPU time results are documented in Fig.~\ref{fig:timecost}, which indicate that the computational complexity of QEM is of~$\mathcal{O}(N)$, which allows the calculation for the interaction between more than $10^4$ particles within $1s$ on a single core.

\begin{figure}[htbp!]
    \centering
    \includegraphics[width = 0.625\linewidth]{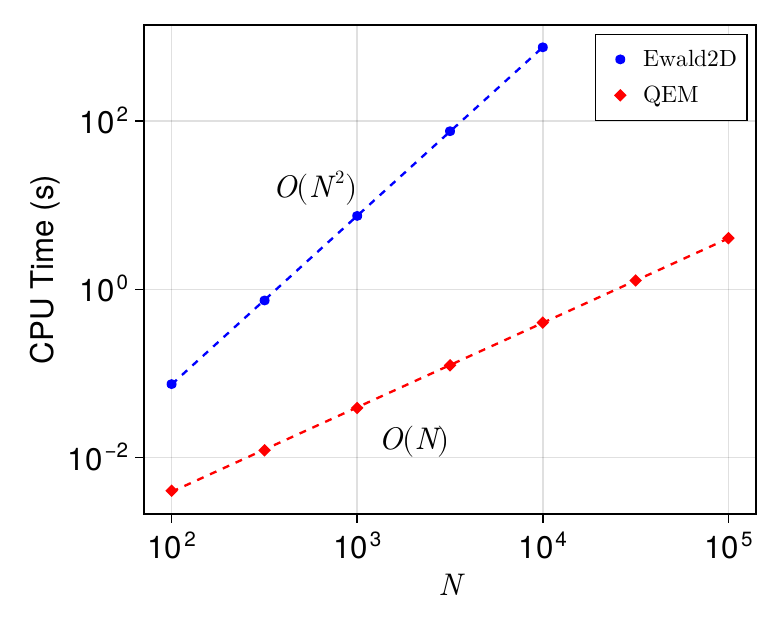}
    \caption{
    CPU time costs for computing the electrostatic interactions between $N$ particles.
    Blue and red dots indicate CPU time costs of Ewald2D method and QEM, respectively. Dashed lines are fitted scaling curves of $\mathcal O(N^2)$ and $\mathcal O(N)$.
    }
    \label{fig:timecost}
\end{figure}

\subsubsection{MD simulations of asymmetric systems}

Now we consider quasi-2D systems with asymmetry.
Specifically, two asymmetric scenarios will be considered here, namely, scenario I: asymmetry in the dielectric mismatches; and scenario II: charge asymmetry of cations and anions.
In both scenarios, the thickness of the system is set as~$L_z = 10 \tau_0$, so that~$L_{x/y} / L_z = 10$.

\textbf{Scenario I: effect of dielectric asymmetry.}
In this scenario, we investigate how asymmetric dielectric mismatches affect equilibrium ionic density distributions.
We set $\gamma_d = -0.95$ and $\gamma_u = 0.95$ for a system containing 218 cations ($+e_0$) and 218 anions ($-e_0$).
Figure~\ref{fig:non_sym} shows the resulting $z$-density profiles.
Both cations and anions are attracted to the conductive lower interface and repelled by the insulating upper interface, creating a strongly biased ionic distribution.

\begin{figure}[htbp!]
    \centering
    \includegraphics[width = 0.625\linewidth]{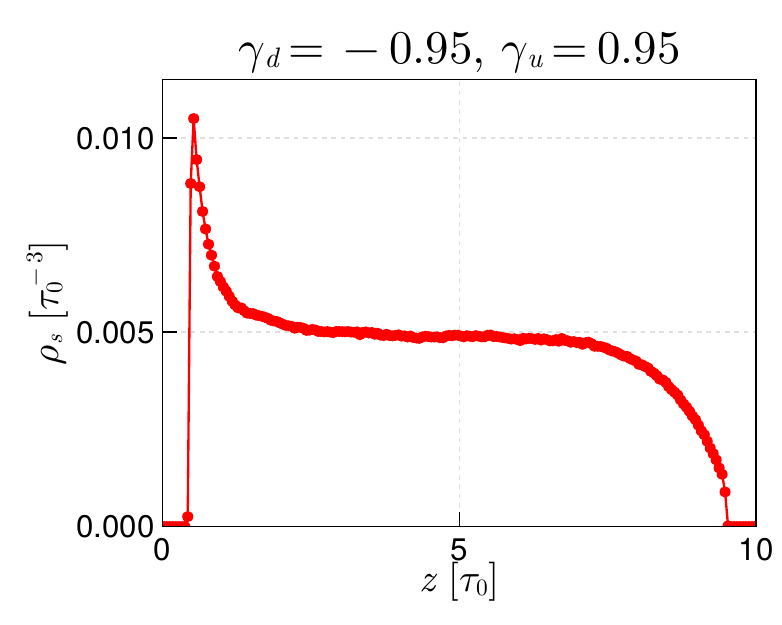}
    \caption{
        Ionic density in~$z$ for symmetric electrolytes containing 218 cations and 218 anions, which are confined by neutral dielectric interfaces with $\gamma_u = 0.95$,~$\gamma_d = -0.95$ at~$z = 0$ and~$10\tau_0$, respectively. 
        Due to charge symmetry of cations and anions, only the cation density is shown.
    }
    \label{fig:non_sym}
\end{figure}

\textbf{Scenario II: effect of charge asymmetry.}
Next, we examine the effect of charge asymmetry using a $3:1$ salt solution.
We set $\gamma_u = \gamma_d$ to be $0.95$ or $-0.95$ and simulate a system with 109 trivalent cations and 327 monovalent anions (maintaining 436 total particles).
For conductor-like interfaces ($\gamma = -0.95$, Fig.~\ref{fig:salt3-1}(a)), both species are attracted to the walls. However, the multivalent cations experience a stronger attraction, leading to a higher concentration near the interface.
This accumulation effectively renders the neutral interface positively charged from a far-field perspective—a manifestation of ``charge renormalization''~\cite{trizac2002simple}.
Consequently, a secondary peak in the anion density appears at approximately $1.5\tau_0$, indicating charge stratification. This effect is known to be amplified near strongly charged surfaces or with larger ions~\cite{li2013ionic}.
Conversely, for insulator-like interfaces ($\gamma = 0.95$, Fig.~\ref{fig:salt3-1}(b)), both ion types are repelled from the walls and concentrate in the channel center.
Due to stronger repulsion forces, the multivalent cations are confined more strongly to the center ($z=0$), creating an effectively positive central layer surrounded by a diffuse cloud of anions.
This results in the formation of an electrical double layer (EDL) at the channel center, driven by the interplay between dielectric boundary effects and many-body electrostatic interactions.
%To conclude, the above discussed scenarios 

\begin{figure}[htbp!]
  \centering
  \includegraphics[width = \linewidth]{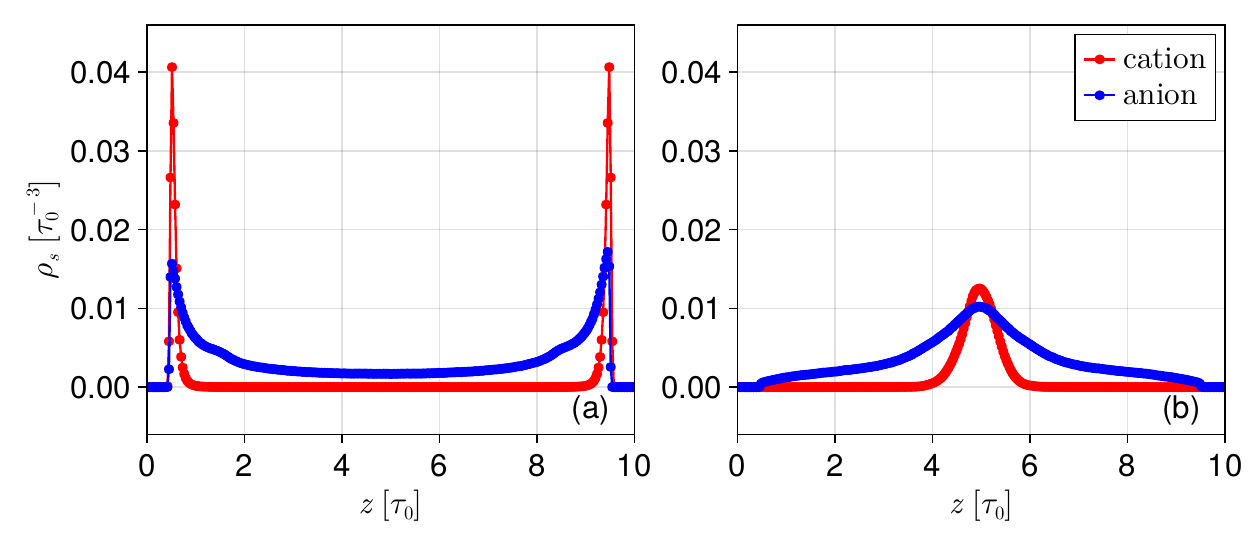}

  \caption{
    Cation and anion densities in~$z$ for a $3:1$ salt with $109$ trivalent cations and $327$ monovalent anions confined by neutral dielectric substrates with both $\gamma_u$ and $\gamma_d$ set to (a) $-0.95$ and (b) $0.95$.
    Here
    red and blue lines indicate densities for cations and anions, respectively.
  }
    \label{fig:salt3-1}
\end{figure}

%These cases indicates that Quasi-Ewald method can also provide significant advantages in handling systems with particle asymmetry, thinness, and strong polarization effects.

\section{Conclusion}
\label{sec:concl}

In summary, we have developed the quasi-Ewald method (QEM) for fast computation of quasi-2D nanoconfined electrostatics. 
Our method is validated to be efficient and accurate, and can be used for investigating various dielectric confinement effects in quasi-2D charged many-body systems. In the future, we will consider extending the QEM to quasi-2D systems confined by metamaterials, where the dielectric contrast $\gamma$ can be extended onto the whole complex plane (except countable infinite poles)~\cite{yu2018plasmonic}. 
Notably, ICM will diverge in such scenarios, but we expect QEM would still work since it does not rely on ICM.
Furthermore, we plan to apply the QEM to real applications, such as simulation studies of bio-membranes and battery-electrolyte systems, where confinement effects
can significantly influence their properties. We finally note that the variance bound proven in Thm.~\ref{thm:unbaised} is based on the DH theory, which becomes inaccurate for strongly-correlated ionic systems~\cite{fisher1993criticality,avni2022conductivity}; obtaining a tight variance bound for such challenging scenarios is reserved for our future investigation.

% \section*{Acknowledgments}

% This work is supported by the Natural Science Foundation of China Grant Nos. 12201146 (Z. G.) and 12401667~(Y. L.); Natural Science Foundation of Guangdong (Grant No.2023A1515012197); and Guangzhou-HKUST(GZ)
% joint research project (Grant No.2023A03J0003 and 2024A03J0606).

\appendix

\section{Proof of Theorem 3.3}

%%%%%%%%%%%%%%%%%%%%%%%%%%%%%%%%%%%%%%%%%%%%%%%%%%%%%%%%%%%%
%%%%%%%%%%%%%%%%%%%%%%%%%%%%%%%%%%%%%%%%%%%%%%%%%%%%%%%%%%%%
%%%%%%%%%%%%%%%%%%%%%%%%%%%%%%%%%%%%%%%%%%%%%%%%%%%%%%%%%%%%
%%%%%%%%%%%%%%%%%%%%%%%%%%%%%%%%%%%%%%%%%%%%%%%%%%%%%%%%%%%%
%%%%%%%%%%%%%%%%%%%%%%%%%%%%%%%%%%%%%%%%%%%%%%%%%%%%%%%%%%%% 
   As we apply quasi-2D Fourier transform  
   on both sides of the Poisson's euqation with dielectric interface conditions:
%%%%%%%%%%%%%%%%%%%%%%%%%%%%%%%%%%%%%%%
%%%%%%%%%%
\begin{equation*} 
    \left \{
    \begin{array}{ll}
        - \grad_{\V{r}} \cdot\left[ \epsilon(\V{r}) \grad_{\V{r}} G(\V{r};\V{r}^\prime) \right] = \delta (\V{r} - \V{r}^\prime), & \text{for}~\V r \in \mathbb{R}^3 \;, \\
%%%%%%%%%%%%%%%%%%%%%%%%%%%        
        G(\V{r};\V{r}^\prime) |_{-} = G(\V{r};\V{r}^\prime) |_{+}, & \text{on}~\partial\Omega_{\mrm{c}}\;, \\
%%%%%%%%%%%%%%%%%%%%%%%%%%%        
        \epsilon_{\mrm{c}} \partial_{z} G(\V{r};\V{r}^\prime) |_{-} = \epsilon_{\mrm{u}} \partial_{z} G(\V{r};\V{r}^\prime) |_{+}, & \text{on}~\partial \Omega_{\mrm{c}} \cap \partial \Omega_{\mrm{u}}\;,     \\
%%%%%%%%%%%%%%%%%%%%%%%%%%%            
\epsilon_{\mrm{c}} \partial_{z} G(\V{r};\V{r}^\prime) |_{-} = \epsilon_{\mrm{d}} \partial_{z} G(\V{r};\V{r}^\prime) |_{+}, & \text{on}~\partial \Omega_{\mrm{c}} \cap \partial \Omega_{\mrm{d}}\;,\\
%%%%%%%%%%%%%%%%%%%%%%%%%%%  
G(\V{r};\V{r}^\prime) \to 0, & \text{as}~{r} \to \infty\;,
\end{array}\right.
\end{equation*}
%%%%%%%%%%%%%%%%%%%%%%%%%%%%%%%%%%%%%%%%%%%%
%%%%%%%%%%%%%%%%%%%%%
%%%%%%%%%%%%%%%%%%%%%%%%%%%%%%%%%%%%%%%%%%%%%%%%%%%%%%%%%%%%
% \begin{equation*} 
%      \left \{
%     \begin{array}{ll}
%         - \grad_{\V{r}} \cdot\left[ \epsilon(\V{r}) \grad_{\V{r}} G(\V{r},~\vr_0) \right] = \delta (\V{r} - \vr_0), & \V r \in \mathbb{R}^3 \;, \\
% %%%%%%%%%%%%%%%%%%%%%%%%%%%        
%         G(\V{r},~\vr_0) |_{-} = G(\V{r},~\vr_0) |_{+}, & \text{on}~\partial\Omega_{\mrm{c}}\;, \\
% %%%%%%%%%%%%%%%%%%%%%%%%%%%        
%         \epsilon_{\mrm{c}} \partial_{z} G(\V{r},~\vr_0) |_{-} = \epsilon_{\mrm{u}} \partial_{z} G(\V{r},~\vr_0) |_{+}, & \text{on}~\partial \Omega_{\mrm{c}} \cap \partial \Omega_{\mrm{u}}\;,     \\
% %%%%%%%%%%%%%%%%%%%%%%%%%%%            
% \epsilon_{\mrm{c}} \partial_{z} G(\V{r},~\vr_0) |_{-} = \epsilon_{\mrm{d}} \partial_{z} G(\V{r},~\vr_0) |_{+}, & \text{on}~\partial \Omega_{\mrm{c}} \cap \partial \Omega_{\mrm{d}}\;,\\
% %%%%%%%%%%%%%%%%%%%%%%%%%%%  
% G(\V{r},~\vr_0) \to 0, & \text{as}~\Norm{\vr} \to \infty\;.
% \end{array}\right.
% \end{equation*}
% %%%%%%%%%%%%%%%%%%%%%%%%%%%%%%%%%%%%%%%%%%%%%%%%%%%%%%%%%%%%
 then for $\vk\neq \vzero$, we have
%%%%%%%%%%%%%%%%%%%%%%%%%%%%%%%%%%%%%%%%%%%%%%%%%%%%%%%%%%%%
    \begin{equation*} 
        \begin{split}
            \frac{\partial^2 \Hat{G}_{\mrm{c}}(\vk, z; z_0)}{{\partial z}^2} - k^2 \Hat{G}_{\mrm{c}}(\vk, z; z_0) &= - \frac{ e^{-\mathrm{i} \V{k} \cdot \V{\rho}_0}\delta (z- z_0)}{\epsilon_{\mrm{c}}},~z \in [0, L_z],\\
            \frac{\partial^2 \hat{G}_{\mrm{u}}(\vk, z; z_0)}{{\partial z}^2} - k^2 \hat{G}_{\mrm{u}}(\vk, z; z_0) &= 0,~z > L_z,\\
            \frac{\partial^2 \hat{G}_{\mrm{d}}(\vk, z; z_0)}{{\partial z}^2} - k^2 \hat{G}_{\mrm{d}}(\vk, z; z_0) &= 0,~z < 0,
        \end{split}
    \end{equation*}
%%%%%%%%%%%%%%%%%%%%%%%%%%%%%%%%%%%%%%%%%%%%%%%%%%%%%%%%%%%%
    with the boundary and interface conditions satisfying
%%%%%%%%%%%%%%%%%%%%%%%%%%%%%%%%%%%%%%%%%%%%%%%%%%%%%%%%%%%%    
    \begin{equation*}
        \begin{split}
            \Hat{G}_{\mrm{c}}(\vk, 0; z_0) &= \Hat{G}_{\mrm{d}}(\vk, 0; z_0)\;,\\
            \hat{G}_{\mrm{c}}(\vk, L_z; z_0) &= \hat{G}_{\mrm{u}}(\vk, L_z; z_0)\;,\\
            \epsilon_{\mrm{c}} \partial_z\Hat{G}_{\mrm{c}}(\vk, 0; z_0) &= \epsilon_{\mrm{d}} \partial_z\Hat{G}_{\mrm{d}}(\vk, 0; z_0)\;,\\
            \epsilon_{\mrm{c}} \partial_z\Hat{G}_{\mrm{c}}(\vk, L_z; z_0) &= \epsilon_{\mrm{u}} \partial_z\Hat{G}_{\mrm{u}}(\vk, L_z; z_0)\;,\\
            \lim_{z \to \infty} \Hat{G}_{\mrm{u}}(\vk, z; z_0) &= \lim_{z \to -\infty}\Hat{G}_{\mrm{d}}(\vk, z; z_0)  = 0\;.
        \end{split}
    \end{equation*}
%%%%%%%%%%%%%%%%%%%%%%%%%%%%%%%%%%%%%%%%%%%%%%%%%%%%%%%%%%%%
Clearly, $\Hat{G}_{\mrm{u}}(\vk, z; z_0)$ and $\Hat{G}_{\mrm{d}}(\vk, z; z_0)$ should take the following forms to satisfy the infinite boundary conditions 
%%%%%%%%%%%%%%%%%%%%%%%%%%%%%%%%%%%%%%%%%%%%%%%%%%%%%%%%%%%%
\begin{align*}
    \Hat{G}_{\mrm{u}}(\vk, z; z_0)&=C_{\mrm{u}}(z_0)e^{-kz},~z > L_z\;,\\
     \Hat{G}_{\mrm{d}}(\vk, z; z_0)&=C_{\mrm{d}}(z_0)e^{kz},~~~~z < 0\;.
\end{align*}
%%%%%%%%%%%%%%%%%%%%%%%%%%%%%%%%%%%%%%%%%%%%%%%%%%%%%%%%%%%%
Therefore, we have 
%%%%%%%%%%%%%%%%%%%%%%%%%%%%%%%%%%%%%%%%%%%%%%%%%%%%%%%%%%%%
\begin{align*}
    \partial_z\Hat{G}_{\mrm{u}}(\vk, z; z_0)&=-k\Hat{G}_{\mrm{u}}(\vk, z; z_0),~z > L_z\;,\\
     \partial_z\Hat{G}_{\mrm{d}}(\vk, z; z_0)&=k\Hat{G}_{\mrm{d}}(\vk, z; z_0),~z < 0\;.
\end{align*}
%%%%%%%%%%%%%%%%%%%%%%%%%%%%%%%%%%%%%%%%%%%%%%%%%%%%%%%%%%%%
Then via Dirichlet-to-Neumann map, we obtain a new set of boundary value problems for $\Hat{G}_{\mrm{c}}(\vk, z; z_0)$ defined on the bounded domain $[0, L_z]$ only: 
%%%%%%%%%%%%%%%%%%%%%%%%%%%%%%%%%%%%%%%%%%%%%%%%%%%%%%%%%%%%
\begin{align}\label{eq:dtnbc1}
 \epsilon_{\mrm{c}}\partial_z\Hat{G}_{\mrm{c}}(\vk, 0; z_0)&=\epsilon_{\mrm{d}} \partial_z\Hat{G}_{\mrm{d}}(\vk, 0; z_0)=k\epsilon_{\mrm{d}} \Hat{G}_{\mrm{d}}(\vk, 0; z_0)=k\epsilon_{\mrm{d}} \Hat{G}_{\mrm{c}}(\vk, 0; z_0)\;,\\
  \epsilon_{\mrm{c}}\partial_z\Hat{G}_{\mrm{c}}(\vk, L_z; z_0)&=\epsilon_{\mrm{u}} \partial_z\Hat{G}_{\mrm{u}}(\vk, L_z; z_0)=-k\epsilon_{\mrm{u}} \Hat{G}_{\mrm{u}}(\vk, L_z; z_0)=-k\epsilon_{\mrm{u}} \Hat{G}_{\mrm{c}}(\vk, L_z; z_0)\;.\label{eq:dtnbc2}
\end{align} 
%%%%%%%%%%%%%%%%%%%%%%%%%%%%%%%%%%%%%%%%%%%%%%%%%%%%%%%%%%%%
Finally, since $\Hat{G}_{\mrm{c}}(\vk, z; z_0)$ takes the form
%%%%%%%%%%%%%%%%%%%%%%%%%%%%%%%%%%%%%%%%%%%%%%%%%%%%%%%%%%%%
\begin{equation*}
    \Hat{G}_{\mrm{c}}(\vk, z; z_0)=\left \{
            \begin{array}{cc}
                \begin{aligned}
                    & Ae^{kz}+Be^{-kz}\;,
                \end{aligned} & z_0<z<L_z, \\
                \begin{aligned}
                 & Ee^{kz}+Fe^{-kz}\;,
                \end{aligned} & 0<z<z_0,
            \end{array}
        \right. \;
\end{equation*}
%%%%%%%%%%%%%%%%%%%%%%%%%%%%%%%%%%%%%%%%%%%%%%%%%%%%%%%%%%%%
by the continuity and jump condition of $\Hat{G}_{\mrm{c}}$ and $\partial_z\Hat{G}_{\mrm{c}}$ at $z=z_0$, respectively; as well as the new BCs at $z=0$ and $z=L_z$ (Eqs.~\eqref{eq:dtnbc1} and~\eqref{eq:dtnbc2}),
one can solve for the unknown coefficients $A$, $B$, $E$, $F$.
 Analogously, one can solve for the case $\vk=\vzero$ and  we omit this for brevity.
%%%%%%%%%%%%%%%%%%%%%%%%%%%%%%%%%%%%%%%%%%%%%%%%%%%%%%%%%%%%
%%%%%%%%%%%%%%%%%%%%%%%%%%%%%%%%%%%%%%%%%%%%%%%%%%%%%%%%%%%% 3

\section{Proof of Theorem 4.1}
For $\Delta \mathrm{I}_b(M)$,  as an upper bound on $\Norm{\fJ_0(r)}$ yields~\cite{landau2000bessel,olenko2006upper},
%%%%%%%%%%%%%%%%%%%%%%%%%%%%%%%%%%%%%%%%%%%%%%%%%%%%%%%%%%%%
\begin{equation*}
   \Norm{\fJ_0(r)}\leq \sqrt{\frac{2}{\pi}} \frac{1}{\sqrt{r}}\;,~\forall r>0\;, 
\end{equation*}
%%%%%%%%%%%%%%%%%%%%%%%%%%%%%%%%%%%%%%%%%%%%%%%%%%%%%%%%%%%%
thus we have 
%%%%%%%%%%%%%%%%%%%%%%%%%%%%%%%%%%%%%%%%%%%%%%%%%%%%%%%%%%%%
\begin{align*}
   & \Norm{\int_M^{+\infty}  e^{-2kL_z}e^{-ak} \fJ_0(\rho k)\D k} 
\leq \int_M^{+\infty}  e^{-2kL_z}\sqrt{\frac{2}{\pi}} \frac{1}{\sqrt{\rho k}} \D k=\frac{1}{\sqrt{\rho L_z}}\mrm{erfc}(\sqrt{2L_zM})\;.
\end{align*}
%%%%%%%%%%%%%%%%%%%%%%%%%%%%%%%%%%%%%%%%%%%%%%%%%%%%%%%%%%%%
And for $\Delta \mathrm{II}(M)$, with the property that $\Norm{\fJ_0(r)}\leq 1$, 
%%%%%%%%%%%%%%%%%%%%%%%%%%%%%%%%%%%%%%%%%%%%%%%%%%%%%%%%%%%%
\begin{align*}
 \Norm{ \int_M^{+\infty}  e^{-\frac{k^2}{4\alpha}}  e^{-ak} \fJ_0(\rho k)\D k}  
  \leq &\int_M^{+\infty} e^{-\frac{k^2}{4\alpha}} \D k=\sqrt{\pi\alpha} \mrm{erfc}\left(\frac{M}{2\sqrt{\alpha}}\right)\;,
\end{align*}
%%%%%%%%%%%%%%%%%%%%%%%%%%%%%%%%%%%%%%%%%%%%%%%%%%%%%%%%%%%%
and we finish our proof.

\section{Proof of Lemma 5.2}
Under the DH approximation, one is able to estimate functions associated with the $i$-th particle in the form:
\begin{equation}
	\fG(\bm{r}_i)=\sum_{j\neq i}q_{j}e^{\mathrm{i} \bm{k}\cdot\bm{\rho}_{ij}}f(z_{ij}),
\end{equation}
where $|f(z_{ij})|$ is bounded by a constant $C_f$ independent of $z_{ij}$. The DH theory considers the simplest model of an electrolyte solution confined to the simulation cell, where all $N$ ions are idealized as hard spheres of diameter $r_{a}$ carrying charge $\pm q$ at their centers. The charge neutrality condition requires that $N_+=N_-=N/2$. Let us fix one ion of charge $+q$ at the origin $r=0$ and consider the distribution of the other ions around it.

In the region $0<r\leq r_{a}$, the electrostatic potential $\phi(\bm{r})$ satisfies the Laplace equation $-\Delta\phi(\bm{r})=0$. For $r\geq r_{a}$, the charge density of each species is described by the Boltzmann distribution $\rho_{\pm}(\bm{r})=\pm qe^{\mp\beta q\phi(\bm{r})}\rho_r/2$ with number density $\rho_r=N/V$. In this region, the electrostatic potential satisfies the linearized Poisson-Boltzmann equation~\cite{levin2002electrostatic}:
\begin{equation}
	-\Delta \phi(\bm{r})=2\pi\left[q \rho_r e^{-\beta q\phi(\bm{r})}-q\rho_r e^{+\beta q\phi(\bm{r})}\right]\approx -4\pi \beta q^2\rho_r\phi(\bm{r}),
\end{equation}
and its solution is given by
\begin{equation}
	\phi(\bm{r})=\begin{cases}
		\dfrac{q}{4\pi r}-\dfrac{q\kappa}{4\pi (1+\kappa a)},& r<r_{a},\\[1em]
		\dfrac{qe^{\kappa a}e^{-\kappa r}}{4\pi r(1+\kappa a)},&r\geq r_{a},
	\end{cases}
\end{equation}
where $\kappa=\sqrt{\beta q^2\rho}$ denotes the inverse of Debye length $\lambda_{\text{D}}$. By this definition, the net charge density for $r>r_{a}$ is $\rho_>(\bm{r})=-\kappa^2\phi(\bm{r})$. Let us fix $\bm{r}_i$ at the origin. Given these considerations, for $r\geq r_a$, one obtains the following estimate:
\begin{equation}\label{eq::E.4}
	\begin{aligned}
		|\fG(\bm{r}_i)|&\approx \left|\int_{\mathbb{R}^3\backslash B(\bm{r}_{i}; r_a)}\rho_>(\bm{r})e^{-\mathrm{i} \bm{k}\cdot\bm{\rho}}f(z)d\bm{r}\right|\\
		&\leq \frac{q_iC_fe^{\kappa a}}{4\pi(1+\kappa a)}\int_{a}^{\infty}\frac{e^{-\kappa r}}{r}4\pi r^2dr\\
		&=q_iC_f\lambda_{\text{D}}^2\;,
	\end{aligned}
\end{equation}
where $ B(\bm{r}_{i}; r_a)$ stands for the ball of radius $r_a$ centered at $\bm{r}_{i}.$

It is noteworthy that upper bound Eq.~\eqref{eq::E.4} is derived under the continuum approximation. In the presence of surface charges, the charge distribution along the $z$-direction may lack spatial uniformity. However, due to the confinement of particle distribution between two parallel plates, the integral in Eq.~\eqref{eq::E.4} along the $z$-direction remains bounded. An upper bound in the form of $|\fG(\bm{r}_i)|\leq C_s C_{f}q_i$ can still be expected, where $C_s$ is a constant related to the thermodynamic properties of the system. The above proof can also be found in Appendix F of Ref.~\cite{gan2025fast}, which is some of the authors' earlier work for quasi-2D Coulomb systems but without considering the sharp dielectric interfaces.

\bibliographystyle{siamplain}
\bibliography{ref}
\end{document}